\documentclass[11pt]{article}
\usepackage{amsmath,amsthm,amscd,amssymb}
\usepackage{color}
\usepackage{graphicx}
\usepackage{graphicx}
\usepackage{authblk}
\usepackage{url}
\usepackage[cal=boondoxo]{mathalfa}
\usepackage{amssymb}
\usepackage{mathabx}  
\usepackage{bbm}
\usepackage{mathtools} 
\mathtoolsset{showonlyrefs}

\usepackage{verbatim}
\usepackage{enumerate, verbatim}
\newcommand{\N}{{\mathbb{N}}}

\newcommand{\R}{{\mathbb{R}}}

\newcommand{\al}{\alpha}

\newcommand{\ga}{\gamma}

\newcommand{\ka}{\kappa}

\newcommand{\om}{\omega}
\newcommand{\Om}{\Omega}

\newcommand{\beq}{\begin{equation}}
\newcommand{\eeq}{\end{equation}}
\newcommand{\bdm}{\begin{displaymath}}
\newcommand{\edm}{\end{displaymath}}
\newcommand{\ba}{\begin{align}}
\newcommand{\ea}{\end{align}}
\newcommand{\bpf}{\begin{proof}}
\newcommand{\epf}{\end{proof}}
\newcommand{\eps}{\epsilon}

\newcommand{\Fumn}{F^{\mu \nu}}
\newcommand{\Fdmn}{F_{\mu \nu}}
\newcommand{\Furt}{F^{\rho \theta}}
\newcommand{\Fdrt}{F_{\rho \theta}}

\newcommand{\sph}{\mathbb{S}^2}

\allowdisplaybreaks

\newcommand{\calD}{\mathcal{D}}

\newcommand{\cF}{\mathcal{F}}

\newcommand{\calK}{\mathcal{K}}

\newcommand{\cJ}{\mathcal{J}}

\newcommand{\calO}{\mathcal{O}}

\newcommand{\calC}{\mathcal{C}}

\newcommand{\wa}{\Box\,}
\newcommand{\del}{\partial}
\newcommand{\half}{\frac{1}{2}}
\newcommand{\fq}{\mathcal{q}}

\newcommand{\SWL}{\mathfrak{Q}}

\newcommand{\Addresses}{
\bigskip
  \footnotesize

\noindent  V.~Hoang, M. Radosz, A. Harb,  \textsc{Department of Mathematics, University of Texas at San Antonio,
    San Antonio, Texas 78249 (USA)}\par\nopagebreak
\noindent A. DeLeon, \textsc{Department of Physics, University of Texas at San Antonio,
    San Antonio, Texas 78249 (USA)}\par\nopagebreak
\noindent  A. Baza, \textsc{Department of Mechanical Engineering, University of Texas at San Antonio,
    San Antonio, Texas 78249 (USA)}\par\nopagebreak
 \vspace{0.3cm}   
  \noindent \textit{E-mail address}, V.~Hoang: \texttt{duynguyenvu.hoang@utsa.edu}\\
  \textit{E-mail address}, M.~Radosz: \texttt{maria\_radosz@hotmail.com}\\
  \textit{E-mail address}, A.~DeLeon: \texttt{aarond.3044@gmail.com}\\
  \textit{E-mail address}, A.~Harb: \texttt{angel-harb@outlook.com}\\
  \textit{E-mail address}, A.~Baza: \texttt{alan.baza@my.utsa.edu}
}

\newtheorem{theorem}{Theorem}
\newtheorem{proposition}[theorem]{Proposition}
\newtheorem{lemma}[theorem]{Lemma}

\theoremstyle{definition}

\newtheorem{definition}[theorem]{Definition}

\newtheorem{remark}[theorem]{Remark}

\newcounter{assumptions}

\begin{document}

\title{Radiation reaction in higher-order electrodynamics}
\author{Alan Baza, Aaron DeLeon, Angel Harb, Vu Hoang, Maria Radosz}
\maketitle 

\begin{abstract}\noindent
This paper considers the relativistic motion of charged particles coupled with electromagnetic fields in the higher-order theory proposed by Bopp, Land\'e--Thomas, and Podolsky.
We rigorously derive a world-line integral expression for the self-force of the charged particle from a distributional equation for the conservation of four-momentum only. This naturally leads to an equation of motion for charged particles that incorporates a history-dependent self-interaction. We show additionally that the same equation of motion follows from a variational principle for retarded fields. The self-force coincides with an expression proposed by Zayats and Gratus--Perlick--Tucker on the basis of an averaging procedure. 
\end{abstract}


%
%

\section{Introduction}
Finding a consistent and well-posed set of dynamical equations for a system of charged particles and their electromagnetic fields is  a classical problem of relativistic physics. A major source of conceptual and mathematical trouble is the strong divergence of the electromagnetic Maxwell-Lorentz fields at the location of point charges. It causes an infinite energy-momentum of the electromagnetic field of a charged point particle, and an a-priori ill-defined Lorentz self-force. 

These problems led Lorentz, Abraham, and many of their contemporaries to conclude that electrons must have a finite size. This \emph{classical electron theory} had to make arbitrary assumptions about the shape and structure of their charge 
and bare mass distributions. For instance, the charged particles were modeled as balls of finite radius, assumed to be rigid in either in their rest frame (Lorentz) or in the frame of the hypothetical ether (Abraham). Yet on certain space and time scales the motion of such extended particles is independent of their size and shape; see \cite{Spohn2004} for a modern mathematical discussion.

In 1938, Dirac \cite{Dirac} reconsidered the motion of charged point particles in the spirit of Lorentz electrodynamics. He heuristically derived an equation of motion by energy-momentum considerations involving an ad-hoc averaging of the electromagnetic fields around a point singularity. A crucial element for Dirac's derivation is the \emph{negative bare mass renormalization} to cancel the infinite electromagnetic field energy of the point charges. The resulting Abraham-Lorentz-Dirac equation of motion is of third order in the time derivative of the particle's position. As a result, almost all of its solutions exhibit an unphysical run-away effect. This deficiency was ameliorated by Landau--Lifshitz who argued that the triple-derivative term must be treated as a small perturbation of the external forces. They thus arrived at a second order equation of motion which is free of run-away solutions. On certain space and time scales, the Landau-Lifshitz equation can be derived rigorously as an effective approximation to the Abraham-Lorentz equation for finite-size charged particles with positive bare mass (see \cite{Spohn2004}).

Another approach to the problem started in 1933 with a paper by Born \cite{Born33}. He argued that to avoid the infinite field energy problem of point charges in Lorentz electrodynamics one should instead work with the nonlinear theory of the classical electromagnetic field started by Mie \cite{Mie}. This avoids all the questions encountered in classical electron theory, i.e. the size, shape, and the charge and mass distributions of an electron but faces the challenge of choosing the correct nonlinearity for the field equations. Born and Infeld \cite{BornInfeld34} proposed a modification of Born's Lagrangian which yields the Maxwell--Born--Infeld field equations. Coupling these MBI field equations consistently with a classical theory of point-charge motion then results in the \emph{Maxwell-Born--Infeld (MBI) electrodynamics}. 

The MBI field equations are nonlinear and hence difficult to handle. For example, there seems to be no known way to obtain solutions of the field equations for arbitrarily prescribed point charge motions. In order to make progress towards understanding the coupling of point particles to the electromagnetic field, Bopp \cite{BoppA,BoppB}, Land\'e--Thomas \cite{Lande41_part1, Lande41_part2}, and Podolsky \cite{Podo42, PodoSchweb48} 
went into a different direction by proposing linear, but \emph{higher-order derivative} field equations to remove the infinite field energy problems. As with the MBI field equations, these linear field equations can be derived from a Lagrangian, which we term BLTP equations. The linearity of the BLTP field equations allows to solve them in great generality for \emph{prescribed motions} of their point sources. Going further to a dynamical description of point charge motion and their associated fields would now require an equation of motion containing a self-interaction term.

Defining the self-interaction or self-force for charged point particles is a highly non-trivial problem. The self-force problem in BLTP electrodynamics was studied in quite some generality by J. Gratus, W. Tucker and V. Perlick \cite{GratusPerlickTucker15}, following earlier studies in \cite{Lande41_part2} and \cite{Zayats2014}. To arrive at a definition of the self-force, the authors of \cite{GratusPerlickTucker15} invoke a \emph{averaging axiom}: the self field should be a suitable average of the field tensor away from the particle position (see also \cite{PPV}). The averaging has to be defined specifically, since in general the averaging of a discontinuous function could produce arbitrary values. 

The authors of \cite{GratusPerlickTucker15} conclude that
the self-force is given by the following integral in flat space-time:
\begin{align}\label{def_self2}
f^{a}_{\text{self}}(\tau) = 
e^2\ka^2 u_b(\tau) \int_{-\infty}^{\tau} \frac{J_2(\ka D(\tau, \tau')) (R^{a} u^{b} - R^{b} u^{a})}{D^2(\tau, \tau')}~d \tau'.
\end{align}
 Here, $J_2$ is a Bessel function, $u^a$ is the four-velocity, $q^a(\tau)$ the position of the charge, $R^a = q^a(\tau) - q^a(\tau')$ and $D=D(\tau, \tau')$ is 
the Minkowski distance between $q^a(\tau)$ and $q^a(\tau')$. The parameter $\ka > 0$ is \emph{Bopp's wave number}, contained 
in the BLTP field equations. Special cases of the self-force \eqref{def_self2} appeared in \cite{Lande41_part2} and \cite{Zayats2014}.

In this paper, we show that energy-momentum conservation, without additional axioms, naturally leads to the explicit expression (\ref{def_self2}) for the self-force on a moving point particle. The strength of our approach consists in showing that \eqref{def_self2} is in fact a consequence of energy-momentum conservation only. To provide a rigorous framework for point particles, we will regard the energy momentum tensor and its divergence as a distribution on space-time. In addition, we derive the equations of motion for $N$ charged particles involving their self-interaction. In the second part of the paper, we show that the equation of motions also follow from a variational principle that couples particles to their retarded fields from the outset. Finally, we show that the equation of motion for a single particle coupled to its retarded field and an external force has a global solution. 

Our work was inspired by a fundamental breakthrough in the formulation of the joint initial-value problem for 
charged point particles and their electromagnetic fields, achieved by M. Kiessling and A.S. Tahvildar-Zadeh in 
\cite{KiesslingTahvildarZadeh18} (see also \cite{Kiessling1, Kiessling2}), where a system of $N$ charged point particles together with their associated BLTP self-fields is considered. Starting from conservation of energy and momentum for the whole system, the authors of \cite{KiesslingTahvildarZadeh18} formulate the initial-value problem as a fixed-point equation on a space of suitable particle trajectories and prove local in-time well-posedness.

There exists a large body of literature on electrodynamic self-force and we do not attempt to give a comprehensive overview. A detailed discussion of the topic's history can be found in \cite{Spohn2004}. Issues arising from consideration of Maxwell-Lorentz theory without self-forces are discussed in \cite{Deckert}. Other aspects of BLTP electrodynamics and possible implications for the gravitational self-force problem can be found in \cite{Perlick15}. The particle dynamics considered in this paper generally lead to an integro-differential equation for the particle motion involving the past history of the particle. This is different from Wheeler-Feynman electrodynamics \cite{WF45, WF49}, which involves integro-differential equations containing the past \emph{and} future of the particle motion. For a mathematical study of the Wheeler-Feynman equations, we refer the reader to \cite{DeckertFW}. 

%
%

\section{Field equations}


\subsection{Notation and conventions}

\textbf{Space-time conventions.} Our setting will be flat Minkowski space-time $\mathbb{M}^4$ with signature $- +++$, $(x^\alpha)$ referring to events measured with respect to an inertial Lorentz frame. Units are chosen such that $c=1$. The field tensor $F_{a b}$ is related to the fields $E, B$ via
\begin{align*}
F_{i 0} = E_i, ~~ F_{i j} = \eps_{ijk} B_k
\end{align*}
where indices $i,j,k$ run over $1,\ldots, 3$.
Note that $B_i = \frac{1}{2} \eps_{i j k} F_{j k}$. Latin indices $a,b,c,d$ and greek indices run over $0, \ldots, 3$.

$g_{ab}$ and $g^{ab}$ will denote the Minkowski metric throughout. Differential operators are defined by $\del_b = (\del_0, \del_i), \del^b = g^{b a}\del_a$. To enhance the readability of formulas we will sometimes write $v \cdot w = v_\al w^\al$. 

The BLTP field equations read
\begin{align}\label{eq_field}
\begin{split}
(I - \ka^{-2} \wa) \del^{c} F_{c d} &= - 4 \pi j_d\\
\del_{a} F_{bc} + \del_b F_{ca} + \del_{c} F_{a b} &= 0
\end{split}
\end{align}
where $\wa = \del^b \del_b$ and $\ka > 0$ is a fixed parameter throughout. $\ka \to \infty$ recovers the Maxwell-Lorentz equations.
We notice that the field equations arise from variation of the following Lagrangian density:
\begin{align}\label{eq_Lagrangian}
\widehat{\mathcal{L}} = - \frac{1}{16\pi} F_{ab} F^{ab} - \frac{1}{8\pi \kappa^2} \del_c F^{ca} \del^d F_{da} + j_b A^b    
\end{align}
where $F_{cd}=\del_c A_d - \del_{d} A_c$ and $A^d$ is the vector potential.

Writing \eqref{eq_field} out in terms of $E$ and $B$ gives the equations
\begin{align*}
\begin{split}
\nabla \times E = - \dot B, ~~\nabla \times H - 4\pi j = \dot D\\
\nabla \cdot B = 0, ~~\nabla \cdot D = 4\pi \rho\\
D = (I + \ka^{-2}(\del_t^2 - \Delta))E, ~~
H = (I + \ka^{-2}(\del_t^2 - \Delta))B.
\end{split}
\end{align*}
Note that $\wa = \del^b \del_b = -\del_0^2 + \del_i \del^i$. 

In the following, $D(x^\alpha)$ denotes the Lorentz distance, i.e.
$$
D(x^a) = \sqrt{- g_{cd} x^c x^d}
$$
and we generally use the symbol $R^a$ for the expression
$$
R^a = x^a - q^a(\tau)
$$
where $q^a$ refers to a specific particle world-line (see below). Vector- and tensor-valued distributions on $\mathbb{M}^4$ will be identified with collections of distributions $\mathcal{D}'(\R^4)$, where $\mathcal{D}(\R^4)$ denotes the space of compactly supported $C_0^\infty$-functions. For example, a space-time tensor field $T^{cd}$ will be defined by a rule that associates a given Lorentz frame with a collection of 16 component distributions $T^{cd}\in \mathcal{D}'(\R^4)$ such that in under a change of Lorentz frame $(x^\alpha) \rightarrow (\bar x^\alpha)$, the components transform as
\begin{align}\label{eq_transform_tensor}
{\bar T}^{cd} = \Lambda^{c}_{r} \Lambda^{d}_{s} T^{rs}
\end{align}
where $\Lambda^c_r$ denotes the usual Lorentz transformation. Here the meaning of the right-hand side of \eqref{eq_transform_tensor} is to be understood in the following sense via the action
$$
\left(\Lambda^{c}_{r} \Lambda^{d}_{s} T^{rs}\right)[\varphi_{cd}] := T^{rs}[\Lambda^{c}_{r} \Lambda^{d}_{s}\varphi_{rs}] 
$$
on a smooth tensor field $\varphi_{cd}$.

The following definition delineates the class of particle trajectories we consider in this paper:
\begin{definition}\label{def_wl}
A subluminal world-line $\fq$ is the image of a $C^2$-mapping $\lambda \mapsto q^\alpha(\lambda)$ with the following properties:
\begin{enumerate}
    \item $q^\alpha(\lambda)$ is defined for all $\lambda \in \R$ and $\frac{d q^\alpha}{d\lambda}$ is a time-like, future pointing and nonzero for all $\lambda$
    \item For the spatial velocity vector $v = (v^1, v^2, v^3)$ defined by
    \begin{align}\label{def_spat_vel}
    v^i = \left(\frac{dq^0}{d\lambda}\right)^{-1} \frac{dq^i}{d\lambda}
    \end{align}
    the bound 
    \begin{align}\label{eq_bound_vel1}
    \sup_{\lambda \leq A} |v(\lambda)| < 1
    \end{align}
    holds for all $A\in \R$, $|v|$ being the Euclidean norm of $v$.
\end{enumerate}
The class of all subluminal world-lines is denoted by $$\SWL.$$
\end{definition}

For a given world-line $\fq$, we shall use $q^\alpha(\tau)$ to denote a particular representation using proper time $\tau$ as a parameter, i.e. $u^\alpha := \frac{d q^\alpha(\tau)}{d\tau}$ has the property $u^a u_a = -1$ and $u^\alpha$ is the usual four-velocity of the particle. In that case, we also use $a^c$ to denote the four-acceleration
\begin{align*}
a^c(\tau) &=  \ddot{q}^c
\end{align*}
of the particle. Moreover, we use $$q^\alpha(t)$$ for a parameterization with respect to coordinate time $t$.  Recall also that the Lorentz factor is defined by
\begin{align*}
    \gamma(t) = \frac{1}{\sqrt{1-v^2(t)}}
\end{align*}
along the world line. 

\begin{remark} 
\begin{enumerate}
\item[(a)] Equation \eqref{def_spat_vel} defines the spatial velocity with respect to any given Lorentz frame and \eqref{eq_bound_vel1} is equivalent to the statement that for any $t$ there exists a $K=K(t) < 1$ with
\begin{align}\label{eq_bound_vel}
|v(t')|\leq K(t) < 1 \quad (t'\leq t).
\end{align}
By the relativistic addition of velocities, the fact that \eqref{eq_bound_vel} holds in one Lorentz frame implies that a bound of the form \eqref{eq_bound_vel} holds in all Lorentz frames.
\item[(b)] Given a $\fq\in \SWL$ with parametrization $q^\alpha(\tau)$ and a $x^\alpha \in \R^4\setminus \fq$, there exists a unique retarded parameter value $\tau_{ret}=\tau_{ret}(x^\alpha)$ such that
$$
x^\alpha - q^\alpha(\tau_{ret})
$$
is light-like and $x^0 - q^0(\tau_{ret}) > 0$. The retarded position of the particle $q^a(\tau_{ret})$ is the intersection of the world-line with the backwards light-cone with apex at the event $(x^\alpha)$; depending on the context, it will also be useful to write $q^a(\tau_{ret}) = (t_{ret}, q(t_{ret}))$ where $t_{ret}$ is the retarded coordinate time (see Figure \ref{fig1}).
\end{enumerate}
\end{remark}

\begin{figure}
    \centering
    \includegraphics[scale=0.5]{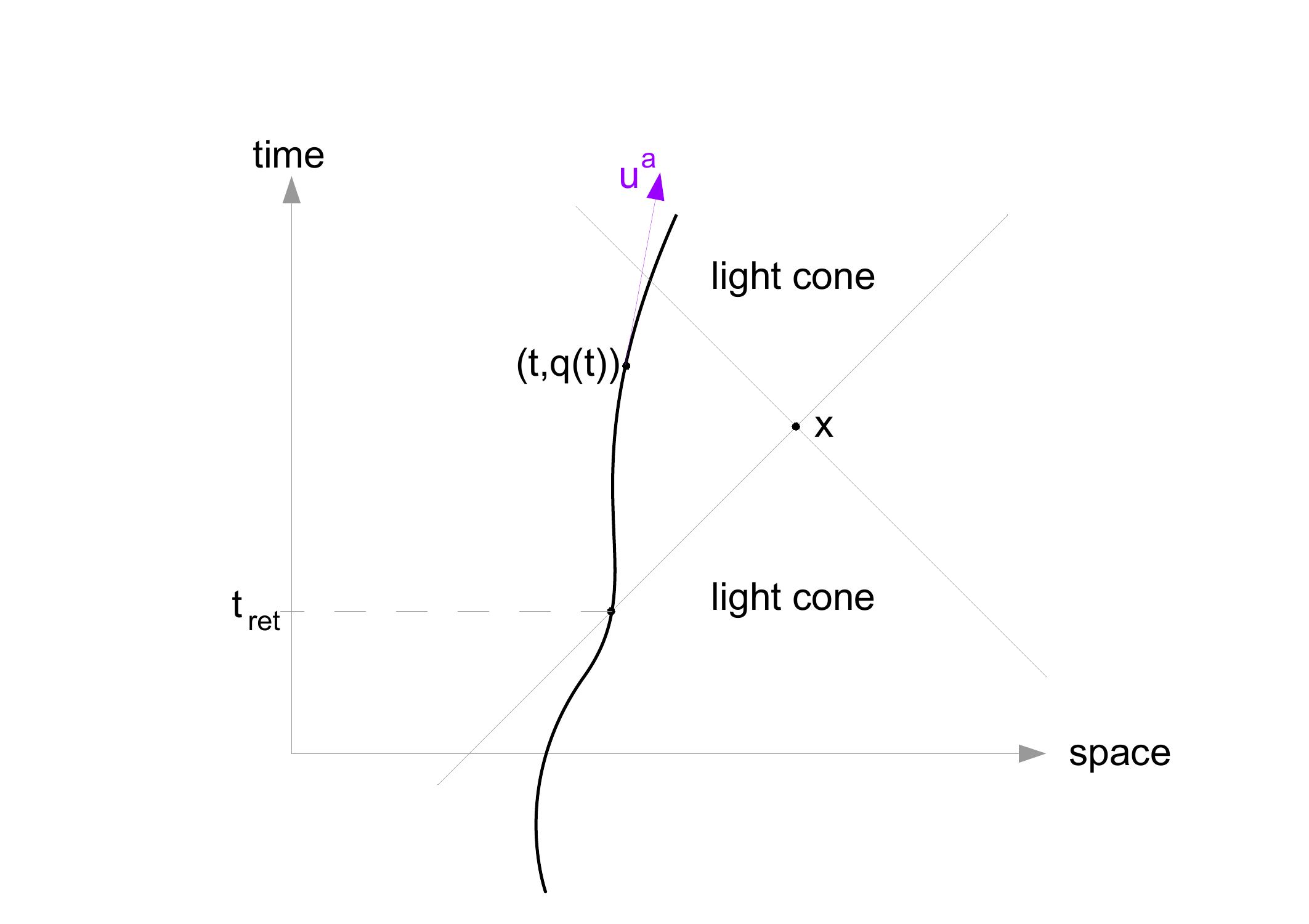}
    \caption{Intersection between particle world-line and backwards light cone}
    \label{fig1}
\end{figure}

For a given $\fq\in \SWL$, we also define curve integrals of the following types:
\begin{align}
\begin{split}\label{world_line_integrals}
    \int_{\fq} f~d\tau &:= \int_{\R} f(q^\alpha(\lambda)) \sqrt{g_{ab}~ \frac{dq^a}{d\lambda} \frac{dq^a}{d\lambda}}d\lambda\\
    \int_{\fq} f~dX^a &:= \int_{\R} f(q^\alpha(\lambda)) \frac{dq^a}{d\lambda} d\lambda\\
    \int_{\fq} F^a ~dX_a &:= \int_{\R} F^a(q^\alpha(\lambda))~g_{ab}\frac{dq^b}{d\lambda} d\lambda.
\end{split}
\end{align}

\noindent \textbf{Conventions for estimates.} Below we often need to bound the components of tensorial expressions. In general, these estimates will be valid for an arbitrary but given Lorentz frame, with constants depending on quantities computed from the velocities and accelerations of particles as referred to that particular Lorentz frame. Inequalities such as
$$
|F_{ca}(x^\alpha)|\leq C(\ldots)
$$
state a bound on all components of the tensor field $F_{ca}$.


\subsection{Retarded solution of the field equations}

In this subsection, we discuss retarded solutions of the field equations corresponding to moving point particles. These are well-known (see \cite{GratusPerlickTucker15, Podo42, PodoSchweb48}), nevertheless it will be useful to collect all formulas for easier reference. 
We start with the case of a single charged particle described by its world line $\fq\in \SWL$.
The current vector is the space-time distribution defined by
\begin{align}\label{eq_current}
j^a = e\int_{\fq} \delta^4(x^a-q^a(\tau))~d X^a= e \int_{-\infty}^\infty  u^a(\tau) \delta^4(x^a-q^a(\tau))~d\tau
\end{align}
where $e$ denotes the particle's charge.

The solution to the BLTP field equation for the vector four-potential $A_a$ can be obtained by the classical observation \cite{Podo42} that $A_a$ can be found by the ansatz $A_a = U_a - W_a$ where $U_a$ satisfies non-massive and $W_a$ satisfies a massive wave equation, i.e. 
\begin{align}\label{eq_diff_two_wav}
\wa U_a = -4\pi j_a, ~~ (\wa - \ka^{2}) W_a = -4\pi j_a.
\end{align}
Here, the potentials $U_a, W_a$ satisfying \eqref{eq_diff_two_wav} are assumed to satisfy the Lorentz gauge condition.

The world-lines we are considering here are subluminal in the sense of Definition \ref{def_wl}, so that for every $x^\alpha$ off the world-line, there exists a unique retarded proper time $\tau_{ret}=\tau(x^\alpha)$ so that
\begin{align}\label{eq_def_tau_ret}
0 = D(x^\al - q^\al(\tau_{ret})) = - g_{\al\beta} (x^\al-q^\al)(x^\beta-q^\beta).
\end{align}
The retarded Green's function for 
\begin{align}\label{eq:klein_gordon}
\wa u - \ka^2 u = - f
\end{align}
reads
\begin{align}\label{eq_Greens}
G^{\ka}(x^{\al}) = \frac{1}{2\pi} \delta (D(x^{\al})^2)\theta(x^0) - \ka \frac{J_1(\ka D(x^{\al}))}{4 \pi D(x^a)}\mathbf{1}_{\{ |x| < x^0 \}}
\end{align}
where $\theta$ denotes the Heaviside function (see \cite{Franklin} for more details).
Using \eqref{eq_Greens} this, we obtain
\begin{align}\label{eq_pot}
A^a(x^\alpha) = e\ka\int_{-\infty}^{\tau_{ret}(x^\alpha)} \frac{J_1(\ka D(x^\mu-q^\mu(\tau)))}{D(x^\mu-q^\mu(\tau))}  u^a(\tau)~d\tau
\end{align}
where $\tau_{ret}=\tau_{ret}(x^\al)$ is the retarded proper time. Observe that the potential $A^a$ can also be written in parametrization-independent form
\begin{align}\label{eq_pot11}
A^a(x^\alpha) = e\ka\int_{\fq\cap \cJ_-(x^\alpha)} \frac{J_1(\ka D(x^\mu-X^\mu))}{D(x^\mu-X^\mu)} ~dX^a
\end{align}
where $\cJ_-$ is the backwards light-cone with apex $(x^\alpha)$.
Before dealing with the convergence of \eqref{eq_pot} and the structure of the field tensors, we define
\begin{align*}
    R^\al = x^\al - q^\al(\tau), ~~ S = u_\ga R^\ga.
\end{align*}
The following relationships can be easily verified:
\begin{align}
\del_c (D(x^\alpha-q^\alpha)) &= - \frac{x_c-q_c}{D(x^\alpha - q^\alpha)}=-\frac{R_c}{D} \nonumber \\
\del_c \tau_{ret}(x^\alpha) &= \left.\frac{x_c-q_c}{u^a (x_a-q_a)}\right|_{ret} = \left.\frac{R_c}{S}\right|_{ret} \nonumber\\
\frac{d}{dz}\left(\frac{J_{n}}{z^n}\right) &= - \frac{J_{n+1}(z)}{z^n}\label{eq_someRel} \\
\del_c R^a(x^\alpha, \tau_{ret}(x^\alpha)) &= \left.\delta^a_c - u^a \frac{R_c}{S}\right|_{ret}\nonumber \\
\del_c R^a(x^\alpha, \tau) &= \delta_c^a\nonumber\\
\del_c S(x^\alpha, \tau_{ret}(x^\alpha)) &= \left.u_c + S^{-1}(1+ a_\ga R^\ga)R_c\right|_{ret}\nonumber
\end{align}
for any positive integer $n$ (see Appendix C for the third line of \eqref{eq_someRel}). In the first line of \eqref{eq_someRel} we refer to the derivative of $D$ with respect to $x^\alpha$ with $\tau$ \emph{held fixed}.
Also, square brackets will denote the following anti-symmetrization operation:
$$
T_{[a}U_{b]} = T_a U_b - T_b U_a.
$$
Note that we don't multiply by the customary factor of $\half$. We also note the following useful inequality
(this can be deduced from standard properties of Bessel functions, see e.g. \cite{Szeg}):
\begin{align}\label{eq_Kiess}
\frac{|J_\nu(x)|}{x^\nu} \leq \frac{C_\nu}{(1+x^2)^{\frac{\nu}{2}+\frac{1}{4}}}
\end{align}
holding for all $x\geq 0,  \nu \in \N_0$.

The following proposition turns out to be useful:
\begin{proposition}\label{prop_LW}
Let $F^{ab}$ be the field tensor for a single particle world-line \eqref{eq_field_tensor}. Then we have
$$
-\ka^{-2} \wa F^{ab} = -F^{ab}+\del^{[a} U^{b]}
$$
where $U^a$ is the Li\'enard-Wiechert potential given by
\begin{align}
U^a = -\left.\frac{e u^a}{u_\alpha R^\alpha}\right|_{ret}. \label{LienardWiechertU}
\end{align}
\end{proposition}
\begin{proof}
This follows by direct calculation from the classical observation \eqref{eq_diff_two_wav} of Podolsky and others that the solution of the BLTP equations can be represented as the difference between two vector potentials. 
\end{proof}

\begin{proposition}\label{prop_conv}
Suppose that $\fq\in \SWL$ is a subluminal world-line and that $\tau\mapsto q^\al(\tau)$ denotes a parametrization with respect to proper time. 
Then \eqref{eq_pot} converges absolutely for all $x^\alpha$ off the world-line of the particle and defines a $C^2$-function of $x^\alpha$ away from the world-line. The field tensor is given by the absolutely convergent expression 
\begin{align}\label{eq_field_tensor}
F_{ca}(x^\alpha) &= e\ka^2 \left.\frac{ R_{[c}u_{a]}}{2 u^b R_b}\right|_{ret} 
+e \ka^2  \int_{-\infty}^{\tau_{ret}(x^\alpha)}
\frac{J_2(\ka D) R_{[c}u_{a]}}{D^2}d\tau\\
& = - \left.\frac{\ka^2}{2} R_{[c} U_{a]}\right|_{ret} + e \ka^2  \int_{-\infty}^{\tau_{ret}(x^\alpha)}
\frac{J_2(\ka D) R_{[c}u_{a]}}{D^2}d\tau
\end{align}
Moreover, the integral in \eqref{eq_field_tensor} is bounded by a constant depending only on $\ka$ and $K(t_{ret})$, where $t_{ret}$ is the retarded coordinate time of $x^\al$ and $K$ is defined in \eqref{eq_bound_vel}. 
The derivatives of $F^{ab}$ are given by
\begin{align}\label{eq_dev_field_tensor}
\begin{split}
\frac{1}{e\ka^2} \del_c F^{ab} &= \left.\half\left[- S^{-3} (1+a_\alpha R^\alpha) R_c - S^{-2}u_c \right] R^{[a} u^{b]}\right|_{ret} + \left.\half S^{-2} R_c R^{[a} a^{b]}\right|_{ret}\\
&\quad + \left.\half S^{-1} \delta_{c}^{[a} u^{b]}\right|_{ret} + \left.\frac{\ka^2}{8}  S^{-1} R_c R^{[a} u^{b]}\right|_{ret}\\
&\quad + \ka \int_{-\infty}^{\tau_{ret}(x^\alpha)}  \frac{J_3(\ka D)}{D^3} R_c R^{[a} u^{b]}~d\tau
+ \int_{-\infty}^{\tau_{ret}(x^\alpha)} \frac{J_2(\ka D)}{D^2} \delta_c^{[a} u^{b]}~d\tau
\end{split}
 \end{align} 
where all the integrals are absolutely convergent. \eqref{eq_field_tensor} solves the field equations \eqref{eq_field} with distributional right-hand side \eqref{eq_current}.
\end{proposition}
The proof can be found in Appendix A.
\begin{remark}
The question of absolute convergence of \eqref{eq_field_tensor} was first addressed in \cite{GratusPerlickTucker15}, where it was shown that \eqref{eq_field_tensor} is absolutely convergent for all world-lines that are bounded away from the past light cone. A world-line is not bounded away from the light-cone if and only if there exists a sequence of proper times $\tau_k \to -\infty$ such that for some $x^\al$, 
$$
\frac{x^\al - q^\al(\tau_k)}{x^0 - q^0(\tau_k)}
$$
approaches a light-like vector. \cite{GratusPerlickTucker15} also discusses an example of a world-line with diverging integral \eqref{eq_field_tensor}.
Throughout this work, we will impose a more restrictive condition on the world-lines $\fq$, namely that $\fq\in \SWL$. As shown in Proposition \ref{prop_conv}, this implies the convergence of \eqref{eq_field_tensor} and guarantees that $F^{cd}\in C^2$ away from the world-line.
\end{remark}

\begin{remark} \eqref{eq_field_tensor} is the analogue of the field tensor derived from the well--known Li\'{e}nard--Wiechert potentials of Maxwell--Lorentz electrodynamics. We follow the usual convention in the physics literature and call \eqref{eq_field_tensor} the field \emph{associated with the particle with world-line} $\fq$. In particular, when $N$ particles with world-lines $\fq_N$ are present, we will refer to the superposition of the $N$ fields of the type \eqref{eq_field_tensor} as the field produced by the $N$ point charges.  
From the point of view of the field equations \eqref{eq_field}, such an association is to a certain extent arbitrary because one can always add a source-free solution to a field representing $N$ point charges and still obtain a solution with the same source. In this sense the ``field of the $N$ point charges'' is not an unambiguously defined term.
\end{remark}

\begin{remark}
In contrast to Maxwell-Lorentz electrodynamics, the retarded field strengths of a particle in BLTP electrodynamics do not diverge to infinity as $x^a$ approaches the world line of the particle. This, however, does not mean that the limit of the field tensor exists on the world line, since the first term in \eqref{eq_field_tensor} has a strong directional dependence. See the discussion in \cite{GratusPerlickTucker15, Perlick15}.
\end{remark}


\subsection{Retarded light-cone coordinates}

It will useful to use retarded light-cone coordinates \cite{Perlick15}, in combination with a Fermi-propagated frame. First recall the concept of Fermi derivative \cite{HawkingEllisBook}:
\begin{definition}
For a vector field $V^\al$ defined along a world-line $\fq\in \SWL$, the Fermi derivative is given by
\begin{align}
    D_F V^\al = \frac{d V^\al}{d\tau} - (V\cdot \Ddot{q}) u^\al + (V\cdot u) \Ddot{q}^\al
\end{align}
where $V\cdot u = V^\ga u_\ga, V\cdot \Ddot{q} = V^\ga \Ddot{q}_\ga$.
\end{definition}
It has the following property:  if $V^\al, W^\al$ are two vector fields along $q^\al(\tau)$, then:
\begin{align}\label{eq_fermi}
    D_F V = 0, ~~D_F W = 0~\text{implies that $V\cdot W$ is constant along $\fq$}.
\end{align}
For the calculations below it will be useful to introduce a Fermi-transported tetrad. We write
\begin{align}
    e_0^\al(\tau) = u^\al(\tau) 
\end{align}
and define vectors $e_i^\al$ as follows. Pick an arbitrary $\tau_0$ and a positively oriented orthonormal set of vectors 
\begin{align}
    e_1^\al(\tau_0), e_2^\al(\tau_0), e_3^\al(\tau_0)
\end{align}
so that $e_i(\tau_0)$ are orthogonal to $e_0(\tau_0)$. Now define $e_1^\al(\tau), e_2^\al(\tau), e_3^\al(\tau)$ on the entire world-line by Fermi transport, i.e. by solving the differential equation
\begin{align}
    0 = D_F e_i^\al, ~~\text{i.e.}~~\dot{e}_i^\al = (e_i \cdot \ddot{q}) u^{\al} - (e_i \cdot u) \Ddot{q}^\al.
\end{align}
Note that the Fermi-transported vectors remain orthonormal for all $\tau$ as a consequence of \eqref{eq_fermi}. The tetrad is not uniquely defined, but depends on the choice of the $e_i^\al(\tau_0)$ at $\tau_0$. As a useful fact, we record
\begin{align}
e_1^\al e_1^\beta + e_2^\al e_2^\beta + e_3^\al e_3^\beta= g^{\al\beta} + u^\al u^\beta.
\end{align}
Next we define retarded light-cone coordinates $(\tau, r, \theta, \varphi)$ associated to every world-line as follows: $(\tau, r, \theta, \varphi)$ is mapped into Minkowski coordinates $(x^\al)$ via
\begin{align}\label{light_cone_coord}
\begin{split}
    N^\al(\tau, \theta, \varphi) &= \cos\varphi \sin \theta ~e_1^\al(\tau) + \sin\varphi \sin\theta ~e_2^\al(\tau) + \cos \theta ~e_3^\al(\tau)\\
    X^\al(\tau, r, \theta, \varphi) &= q^\al(\tau) + r(u^\al(\tau) + N^\al(\tau, \theta, \varphi)). 
\end{split}
\end{align}
\begin{definition}
A world-tube of size $\eps$ around $q^\al(\tau)$ is defined as 
\begin{align*}
    W_\eps = \{ X^\al(\tau, r', \theta, \varphi) : \tau \in (-\infty, \infty), 0\leq r' \leq \eps, \theta\in [0,\pi], \varphi\in [0, 2\pi] \}.
\end{align*}
For an illustration, see Figure \ref{figWT}.
\begin{figure}[h]
    \centering
    \includegraphics[scale=0.4]{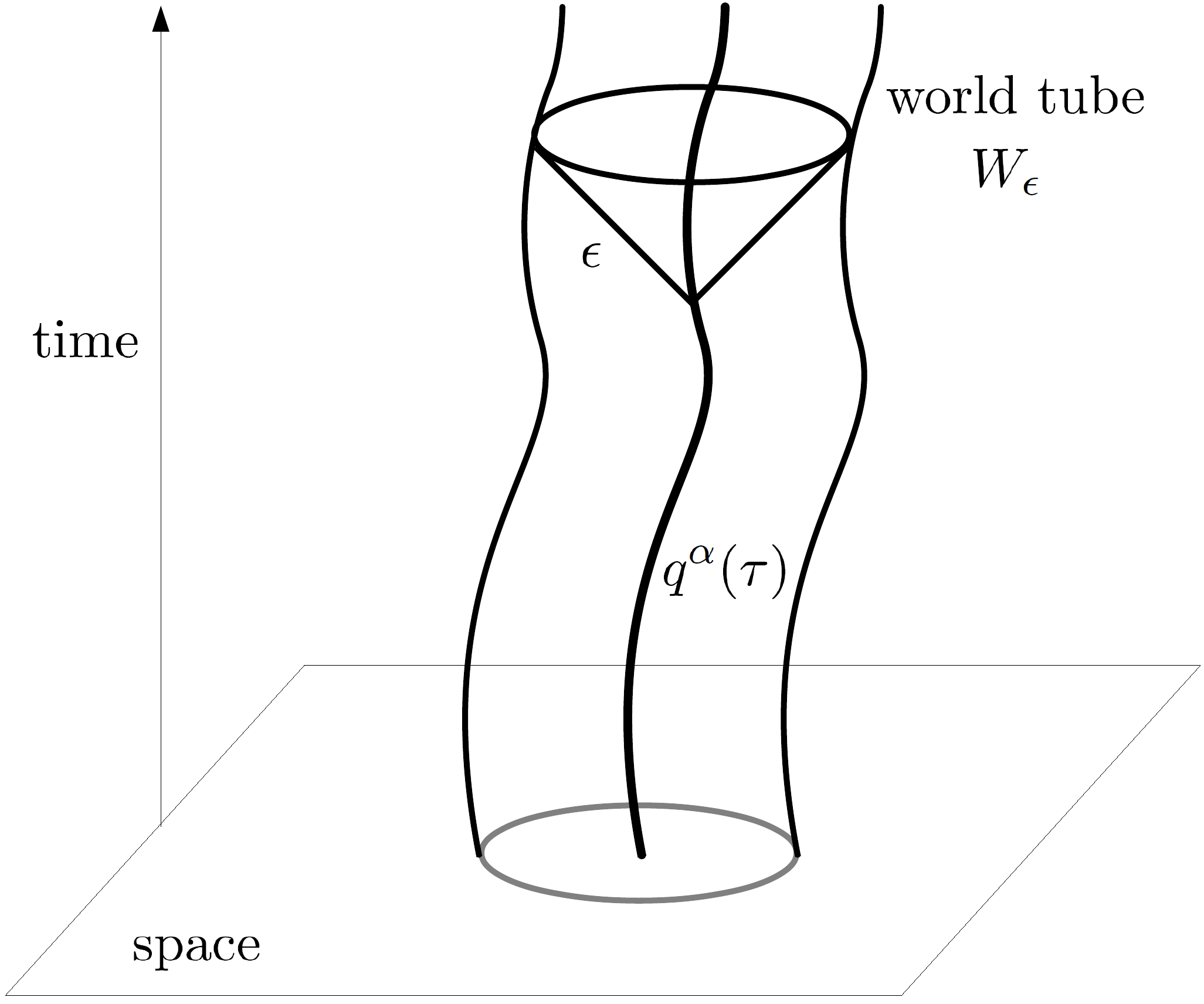}
    \caption{World-tube}
    \label{figWT}
\end{figure}
\end{definition}

Associated to the coordinate transformation $(\tau, r, \theta, \varphi) \mapsto X^\al(\tau, r, \theta, \varphi)$ we have the vector fields $\del_\tau, \del_r, \del_\theta, \del_\varphi$. These form a frame at each point $x^\al\notin \fq$ and are given by
\begin{align}\label{def_vectors_E}
\begin{split}
    \del_\tau & = u + r( \ddot{q} + (N \cdot \ddot{q}) u)\\
    \del_{r} &= u + N\\
    \del_{\theta} &= r \left(\cos \theta \cos \varphi~ e_1 + \cos \theta \sin \varphi~ e_2 - \sin \theta ~e_3\right) \\
    \del_{\varphi} &= r(-\sin\varphi\sin\theta ~e_1 + \cos\varphi\sin\theta~ e_2). \end{split}
\end{align}
\begin{remark} Using the covectors $d\tau, dr, d\theta, d\varphi$, we can express the canonical volume form as
$\sqrt{-g}~ d\tau \wedge dr \wedge d\theta \wedge d\varphi$, where in the $(\del_\tau, \del_r\, \del_\theta, \del_\varphi)$ frame the metric determinant $g$ is given by
\begin{align}
    g = \det{[g(\del_\al, \del_\beta)]} = -r^4 \sin^2 \theta. 
\end{align}
\end{remark}
The following Lemma takes care of integration on $\del W_\eps$:
\begin{lemma}\label{lem_integral_over_world_tube}
On the boundary $\del W_\eps$, we have an outward-oriented three-dimensional surface element $d\Sigma_d$ (see e.g. \cite{Poisson2004}) given by 
\begin{align}\label{eq_surf_elem}
d\Sigma_d =  \om_d\,\eps^2 \sin \theta\, d\tau d\theta d\varphi
\end{align}
where
\begin{align}
    \om_d = \eps (\Ddot{q}\cdot N) u_d + (1+\eps \Ddot{q}\cdot N) N_d
\end{align}
where $N$ is defined in \eqref{light_cone_coord}.
\end{lemma}
\begin{proof}
This is a straightforward computation using the definition of 
\begin{align}
d\Sigma_d = \sqrt{-g}\,[d~\al~\beta~\ga]~\del_\tau X^\al \,\del_\theta X^\beta\, \del_\varphi X^\ga \, d\tau d\theta d\varphi
\end{align}
where $[d~\al~\beta~\ga]$ is the totally antisymmetric permutation symbol (see \cite{Poisson2004}).
\end{proof}


\subsection{Energy-momentum tensor and fields}

The energy-momentum tensor compatible with the BLTP field equations is given by
\begin{align}
\begin{split}\label{def_EM}
-4\pi T^{c d}(F) &:=  g^{c \rho} F_{\rho \mu} F^{\mu d} + \frac{1}{4} g^{cd} \Fdrt \Furt - \ka^{-2}( g^{c\rho} F_{\rho \mu} \wa F^{\mu d} + g^{c \rho} F^{d\mu} \wa F_{\mu \rho}\\
&\quad + g^{c \rho} \del_\mu F^{\mu d} \del^\beta F_{\rho \beta}) - \half \ka^{-2} ( \Fdrt \wa \Furt + \del_\rho \Furt \del^{\beta} F_{\beta \theta})g^{cd}.
\end{split}
\end{align}
This was first given Podolsky \cite{Podo42} in somewhat different notation. The next two propositions describe important properties of $T^{cd}(F)$:

\begin{proposition}\label{prop:EM_tens_cons} Suppose $F^{cd}\in C^2(\R^4)$ solves the BLTP field equations with some current $j^a\in C^1(\R^4)$. Then we have
$$
-\del_d T^{c d}(F) = F^{c d} j_d.
$$
\end{proposition}
The proof is given in Appendix B. In order to exploit the bilinear structure of the energy-momentum tensor, we define the bilinear form
\begin{align}\label{def_EM2}
\begin{split}
-8\pi T^{c d}(F, \widehat{F}) := & g^{c \rho} F_{\rho \mu} \widehat{F}^{\mu d} + \frac{1}{4} g^{cd} \Fdrt \widehat{F}^{\rho\theta} - \ka^{-2}( g^{c\rho} F_{\rho \mu} \wa \widehat{F}^{\mu d} + g^{c \rho} F^{d\mu} \wa \widehat{F}_{\mu \rho}\\
& + g^{c \rho} \del_\mu F^{\mu d} \del^\beta \widehat F_{\rho \beta}) - \half \ka^{-2} ( \Fdrt \wa \widehat F^{\rho \theta} + \del_\rho \Furt \del^{\beta} \widehat F_{\beta \theta})g^{cd} + (F \leftrightarrow \widehat F)
\end{split}
\end{align}
where $(F \leftrightarrow \widehat F)$ indicates the a repetition of the same expression with the roles of $F$ and $\widehat F$ interchanged. By slightly generalizing the calculation in Appendix B, we obtain
\begin{proposition}
For $C^2$-solutions $F, \widehat F$ of the field equations we have
\begin{align}
-\del_d T^{c d}(F, \widehat F) = \half( F^{c d} \widehat{j_d} + \widehat F^{c d} j_d)
\end{align}
where $j_d, \widehat{j_d}$ are the currents corresponding to the fields $F^{cd}, \widehat F^{cd}$.
\end{proposition}

%
%

\section{The principle of Energy-Momentum conservation}


\subsection{Distributional formulation}

In the following, we derive the relativistic equations of motion from conservation of energy and momentum. In particular we derive rigorously the self forces on the particle. For simplicity, we will neglect source-free (incident) fields and assume that particles are neither created nor annihilated, so that the world-lines of the particles extend both into the infinite past and infinite future. The basic setup is as follows: let $q_n^\al(\tau)$ be smooth, non-intersecting world-lines of $N$ particles ($n=1, \ldots, N$), parametrized using proper times. Let $u_n^\al$ denote the four-velocities. We define the \emph{bare mass} energy-momentum tensors:
\begin{align}
T_{kin,n}^{a b} = m_n \int_{-\infty}^\infty u^{a}_n u^{b}_n \delta^{(4)}(x^\alpha-q_n^\alpha(\tau_n))~d\tau_n.
\end{align}
Note that these are distributions with support on the world-lines.

Further let $F^{ab} = \sum_{n=1}^N F_{n}^{a b}$ be the total BLTP field associated to all the particles and let $T_{field}^{cd} = T^{cd}(F)$ be the associated energy-momentum tensor of the field. The principle of energy-momentum conservation is simply
\begin{align}\label{cons}
\del_c\left( \sum_{n=1}^N T_{kin, n}^{ac} + T_{field}^{ac}\right) = 0.
\end{align}
First, it is necessary to show that $T_{field}^{ac}$ is a well-defined distribution. Secondly, we would like to derive the equations of motions rigorously from \eqref{cons}. Before we come to our main result, we need to explain how \eqref{cons} defines the self-field in a natural way. Consider for the moment the case of a single charged particle. 
From Proposition \ref{prop:EM_tens_cons}, we know that 
\begin{align}\label{e1}
-\del_d T^{cd} = F^{c d} j_d
\end{align}
for smooth currents $j_d$. For the fields generated by a single particle, \eqref{e1} does not make classical sense, because $j_d$ is concentrated on the world-line of the particle and $F^{c d}$ does not make sense there because of the first term of \eqref{eq_field_tensor}. However, we will show below that the divergence of $T^{cd}$ is a well-defined distribution with support on the world-line of the particle and that it exactly has the form $e \cF^{ab} u_b$ with an explicitly known field $\cF^{cd}$. It is therefore natural to regard $\cF^{cd}$ as the self-field of the particle.

Our main result is:
\begin{theorem}\label{thm_main}
Suppose $\{q_n^\al(\tau_n)\}_{n=1}^N$ is a collection of non-intersecting world-lines in $\SWL$ and let the field associated to the $n$-th particle be defined by 
\begin{align}\label{eq_field_tensor2}
(F_n)_{ca}(x^\alpha) &= e_n\ka^2 \left.\frac{ R_{[c}u_{a]}}{2 u^b R_b}\right|_{ret} 
+e_n \ka^2  \int_{-\infty}^{\tau_{n,ret}(x^\alpha)}
\frac{J_2(\ka D) R_{[c}u_{a]}}{D^2}d\tau
\end{align}
where $e_n$ is the charge of the $n$-th particle and the integral extends over the world-line $q_n(\tau)$. $R^a, u^a$ refer to the $n$-th particle. Then the following conclusions hold:
\begin{itemize}
\item[(A)]
The energy momentum tensor $T_{field}^{ab}$ of the total retarded field of the $N$ particles is in $L^1_{loc}(\mathbb{M}^4)$ and hence defines a distribution in $\mathcal{D}'(\mathbb{M}^4)$, i.e. acting on compactly supported $C^\infty$-functions $\phi$ via
\begin{align*}
T^{ab}_{field}[\phi_{ab}] = \int_{\R^4} T^{ab}_{field}\phi_{ab}\sqrt{-g}~d^4 x.
\end{align*}
\item[(B)]
Suppose moreover that the energy-momentum equation \eqref{cons} holds in the distributional sense. Then the following equations of motions hold:
\begin{align}\label{eq_motion}
m_n \frac{d^2 q_n^a}{d\tau_n^2} = e_n\left(\mathcal{F}_n^{ab} + \sum_{m\neq n} F_m^{ab}\right)(\dot{q}_n)_b.
\end{align}
Note in particular that the force on the $n$-th particle splits into a self-force $e_n \cF^{ab}u_b$ and the Lorentz-force exerted by all other particles.
The self-field $\cF^{ab}$ is given by the absolutely convergent integral 
\begin{align}\label{eq_self1}
\cF_n^{ca}(\tau) = e_n \ka^2  \int_{-\infty}^{\tau}
\frac{J_2(\ka D) R^{[c}u^{a]}}{D^2}d\tau'
\end{align}
where $D = D(q_n(\tau) - q_n(\tau')), R^a = q_n^a(\tau)-q_n^a(\tau')$.
\end{itemize}
\end{theorem}
\begin{remark}
Theorem \ref{thm_main} consists of two parts: The point of statement $(A)$ is that $T_{field}$ defines a distribution for arbitrary, non-intersecting motions $q^\alpha_n(\tau)$. This allows us to interpret equation \eqref{cons} in the sense of distributions. 
In the second part $(B)$ of the Theorem, we assume that the conservation equation \eqref{cons} holds for the field associated to the $N$ particle world-lines. \emph{As a consequence, the world-lines are no longer arbitrary and necessarily satisfy the equations of motions \eqref{eq_motion}.} 
Note that $\cF^{ab}$ is exactly the field proposed in \cite{GratusPerlickTucker15} as a suitable expression for the self-force.
\end{remark}
\begin{remark}
Results similar to those presented in Theorem \ref{thm_main} hold also for the case of scalar fields $\phi$ with Lagrangian
\begin{align*}
    \mathcal{L} = - \half g^{\al\beta}\del_\al \phi \del_\beta \phi + \frac{1}{2\ka^2} (\wa \phi)^2
\end{align*}
\end{remark}

The proof of our central result Theorem \ref{thm_main} in Section \ref{sec::main_result_pf} requires substantial preparation, to which we turn next. 


\subsection{Properties of self-field energy-momentum tensor}

Our goal is now to study the energy-momentum tensor of the self-field of a single particle moving on an arbitrary world-line $\fq$ satisfying the conditions formulated in Theorem \ref{thm_main}. In the following, $T = T(F)$ will always denote the energy-momentum tensor associated to the field of the world line given by \eqref{eq_field_tensor}. 
The following proposition serves as a technical preparation and describes the asymptotics of the field tensor on the boundary of a world-tube of radius $\eps > 0$. 
\begin{proposition}\label{prop:as1} Let the world-line $\fq\in \SWL$. On $\del W_\eps$, i.e. evaluated at $x^\al = X^\al(\tau, \eps, \theta, \varphi)$ we have as $\eps\to 0^+$
\begin{align}
    R^\al &= \eps (u^\al + N^\al)\label{expr1}\\
    R_{[c} u_{a]} &= \eps N_{[c} u_{a]}\label{expr2}\\
    S &= u^\al R_\al = - \eps\label{expr3}\\
    a^\al R_\al &= \eps a^\al N_\al = \eps (a\cdot N)\label{expr4}\\
    U^a &= \frac{e u^a}{\eps} \label{expr5}\\
    \del^c U^b &= - \frac{e (u^c+N^c)(a^b + (a\cdot N)u^b)}{\eps} - \frac{e u^b N^c}{\eps^2} \label{expr6}\\
    \del^{[c} U^{b]} &= -\frac{e N^{[c} u^{b]}}{\eps^2} + \calO(\eps^{-1})\label{expr7}\\
    F^{ca} &= -\frac{e \ka^2}{2} N^{[c} u^{a]} + \mathcal{F}^{ca}(\tau) + \calO(\eps) \label{expr8}\\
    \frac{1}{e\ka^2} \del_\mu F^{\mu d} &= - \frac{u^d}{\eps} + \calO(1)  
\end{align}
where $U$ is the potential from \eqref{LienardWiechertU} and 
\begin{align}
(\cF)_{ca}(\tau) = e_n \ka^2  \int_{-\infty}^{\tau}
\frac{J_2(\ka D) R_{[c}u_{a]}}{D^2}d\tau'
\end{align}
where $R=q^\al(\tau)-q^\al(\tau')$
Furthermore, 
\begin{align}
    \om_d &= N_d + \calO(\eps)\\
    \cF^{d\mu}\om_d &= \calO(1) \label{expr9}\\
    \del^{[\mu} U^{d]} \om_d &= e \frac{u^\mu}{\eps^2} + \calO(\eps^{-1}) \label{expr10}\\
    F^{d\mu} \om_d &= -\frac{e \ka^2}{2} u^\mu + \cF^{d\mu}N_d + \calO(\eps) \label{expr11}\\
    \del_\mu F^{\mu d} \om_d &= \calO(1).\label{expr12}
\end{align}
All terms on the right-hand sides of the above asymptotic expansions are evaluated at $\tau, \theta, \varphi$. All the $\calO$ terms admit estimates that are uniform in $\tau \theta, \varphi$, as long as $\tau$ varies in a compact interval. 
\end{proposition}
\begin{proof}
As a useful fact, note that $u_\al N^\al = 0$, $N_\al N^\al = 1, u_\al u^\al = -1$ and $u_\al a^\al = 0$. \eqref{expr1}--\eqref{expr5} follow directly from the definitions of $R^a, S, U^a$ and $X^\al(\tau, \eps, \theta, \varphi)$. Also observe that the retarded time $\tau_{ret}$ corresponding to the point $x^\al=X^\al$ is simply $\tau$.
Using \eqref{eq_someRel}, the expressions \eqref{expr6} and \eqref{expr7} can be calculated directly. To show \eqref{expr8}, we note that from \eqref{eq_field_tensor},
\begin{align*}
    F_{ca}(X^\al(\tau, \eps, \theta, \varphi)) = - \frac{\ka^2}{2} e N_{[c} u_{a]} + e \ka^2  \int_{-\infty}^{\tau}
\frac{J_2(\ka \tilde D) \tilde{R}_{[c}u_{a]}}{\tilde{D}^2}d\tau'
\end{align*}
where $\tilde R_c = X_c(\tau, \eps, \theta, \varphi) - q_c(\tau'), \tilde D = D(X^\al(\tau, \eps, \theta, \varphi)-q^\al(\tau'))$. Next, we note that with $R_c = q_c(\tau)-q_c(\tau'), D = D(q^\al(\tau)-q^\al(\tau'))$ we have
\begin{align*}
    \frac{1}{e\ka^2}\left( e\ka^2 \int_{-\infty}^{\tau}
\frac{J_2(\ka \tilde D) \tilde{R}_{[c}u_{a]}}{\tilde{D}^2}~d\tau' - \cF_{ca}(\tau) \right)
= \int_{-\infty}^{\tau}
\left(\frac{J_2(\ka \tilde D)}{\tilde{D}^2} \tilde{R}_{[c} - \frac{J_2(\ka D)}{D^2} R_{[c}\right)u_{a]} ~d\tau'.
\end{align*}
Inserting cross-terms and using inequality \eqref{eq_Kiess}, Proposition \ref{prop_lower_bound_D} in Appendix A and Lemma \ref{lem_besseldifference}, we can estimate the absolute value of this expression by $\eps$ times constant that does not depend on $X^\al$.  

\eqref{expr9}-\eqref{expr11} are consequences of \eqref{expr7} and \eqref{expr8}. To show \eqref{expr12}, we calculate $\del_\mu F^{\mu d}$ by using the representation in Proposition \ref{prop_conv} and apply \eqref{expr1}-\eqref{expr3}. 
\end{proof}

\begin{lemma}[Distributional divergence of the energy-momentum tensor]\label{lemma_distr_div}  Let $F$ be the field tensor of a single world-line $q^\alpha(\tau)$ and let $T = T(F)$ be the corresponding energy-momentum tensor. Then:
\begin{enumerate}
\item[(a)] $T\in \mathcal{D}'(\mathbb{M}^4)$
\item[(b)]
The distributional divergence of the energy-momentum tensor is given by
\begin{align}\label{eq_action_of_divT}
-\del_d T^{c d}[\phi] = e\int_{\fq}  \cF^{c\mu} u_\mu \phi~d\tau
\end{align}
for any smooth test function $\phi$ with compact support on flat space-time. In particular, the distributional divergence of the energy-momentum tensor is supported on the world-line of the particle.
\end{enumerate}
\end{lemma}
\begin{proof}
From Proposition \ref{prop:as1} we see that every term of the energy-momentum tensor \eqref{def_EM} is at most of order $\calO(r^{-2})$ close to the world-line. This implies that $T(F)\in L_{loc}^1(\R^4)$, hence $T(F)$ defines a distribution in $\calD'(\R^4)$.

Recall that the distributional divergence is defined by
\begin{align*}
\del_d T^{c d}[\phi] = - \int_{\R^4} T^{cd} \del_d \phi~d^4 x.
\end{align*}
We surround the world-line by a world tube 
\begin{align}\label{eq_World_tube}
W_\eps = \{ X^\al(\tau, r, \theta, \varphi) : \tau\in (-\infty, \infty), 0 < r \leq \eps, \theta\in [0,2\pi],
\varphi\in [0, 2\pi]\}.
\end{align}
On the boundary $\del W_\eps$, we have an outward-oriented three-dimensional surface element $d\Sigma_d$ (see Lemma \ref{lem_integral_over_world_tube}) given by 
$\Sigma_d = \om_d \eps^2 \sin \theta d\tau d\theta d\varphi$, where
$\om_d =  \eps (a\cdot N) u_d + (1+\eps (a\cdot N)) N_d$. For an illustration, we refer to Figure \ref{figWT}.
Using the Gauss-Stokes theorem and Lemma \ref{lem_integral_over_world_tube} and the fact that the divergence of $T^{cd}$ is zero off the world-line, we get
\begin{align}\label{eq_lim}
\del_d T^{c d}[\phi] = \lim_{\eps\to 0^+} \int_{\del W_\eps} T^{cd} \phi~ d\Sigma_d = 
\lim_{\eps\to 0^+} \eps^2 \int_{-\infty}^\infty \int_{0}^{\pi}\int_0^{2\pi} T^{cd} \phi~ \om_d~\sin\theta ~d\theta~ d\varphi~ d\tau.
\end{align}
The first step to compute the limit \eqref{eq_lim} is to use Proposition \ref{prop_LW} to replace the terms containing wave operators by the Li\'enard-Wiechert potential $U$. Thus we can write 
\begin{align}
-4\pi T^{cd} =: T_1^{cd} + T_2^{cd} 
\end{align}
where 
\begin{align}
\begin{split}
T_1^{cd} &= g^{c\rho}\left[F_{\rho \mu} \del^{[\mu} U^{d]} + F^{d\mu} \del_{[\mu} U_{\rho]}
- \ka^{-2} \del_{\mu} F^{\mu d} \del^{\beta} F_{\rho \beta}\right]+\half g^{cd} \left[F_{\rho \theta} \del^{[\rho} U^{\theta]} - \ka^{-2} \del_{\rho} F^{\rho \theta} \del^{\beta} F_{\beta \theta}\right]
\end{split}\label{eq_T1}
\end{align}
and $T_2^{cd} = - g^{c\rho} F^{d\mu} F_{\mu \rho} - \frac{1}{4} g^{cd} F_{\rho\theta} F^{\rho \theta}$. With Proposition \ref{prop:as1} we see that $T_2^{cd} = \calO(1)$
and hence $T_2$ does not contribute in \eqref{eq_lim} as $\eps\to 0^+$. We now insert into $T_1^{cd}\om_d$ the asymptotic representations of $F$ and $\del_\mu F$ given in Proposition \ref{prop:as1}. 
This results in
\begin{align*}
T_1^{cd}\om_d &= g^{c\rho}\left[\left(\frac{-e\ka^2}{2}N_{[\rho}u_{\mu]}+\cF_{\rho\mu}(\tau)+\calO(\eps)\right)
 \left(\frac{e}{\eps^2} u^\mu + \calO(\eps^{-1})\right) \right. \\
&\quad+\left.  \left(-\frac{e\ka^2}{2} u^\mu + \cF^{d\mu}(\tau)N_d + \calO(\eps)\right)\left(-\frac{e}{\eps^2} N_{[\mu} u_{\rho]} + \calO(\eps^{-1})\right) + e \calO(1) \left(-\frac{u_\rho}{\eps} + \calO(1))\right)\right]\\
&\quad+ \half \left(N^c + \calO(\eps)\right)\left[ \left(-\frac{e\ka^2}{2} N_{[\rho} u_{\theta]} + \cF_{\rho\theta} +\calO(\eps)\right) \left(-\frac{e}{\eps^2} N^{[\rho}u^{\theta]}+\calO(\eps^{-1})\right)\right.\\
&\quad\left.- e\left(\frac{u^\theta}{\eps}  + \calO(1)\right)e \ka^2 \left(\frac{u_\theta}{\eps} + \calO(1)\right) \right].
\end{align*}
After some simplification, we get
\begin{align*}
    T_1^{cd}\om_d = \frac{e}{\eps^2} g^{c\rho}\left[  \cF_{\rho\mu} u^\mu - \cF^{d\mu} N_d N_{[\mu} u_{\rho]}\right] - \frac{e}{2\eps^2} N^c\left[\cF_{\rho\theta} N^{[\rho} u^{\theta]}\right] +\calO(\eps^{-1}) 
\end{align*}
and by using $N_d u^d = 0, u_d u^d = -1$, we can further simplify this to
\begin{align}
    T_1^{cd}\om_d = \frac{e}{\eps^2} \cF^{c\mu} u_\mu + \calO(\eps^{-1}).
\end{align}
Returning to the limit \eqref{eq_lim}, we use $\phi(X^\al(\tau, \eps, \theta, \varphi)) = \phi(q^\al(\tau)) + \calO(\eps)$ to find
\begin{align*}
    &\lim_{\eps\to 0^+} \eps^2 \int_{-\infty}^\infty \int_{0}^{\pi}\int_0^{2\pi} T^{cd} \phi~ \om_d~\sin\theta~ d\theta~ d\varphi~ d\tau \\
    &=\lim_{\eps\to 0^+}  \eps^2 \int_{-\infty}^\infty \int_{0}^{\pi}\int_0^{2\pi} \left(-\frac{e}{4\pi \eps^2} \cF^{c\mu}u_\mu +\calO(\eps^{-1})  \right)\left( \phi(q^\al(\tau)) + \calO(\eps) \right)
    ~\sin\theta ~d\theta~ d\varphi~ d\tau\\
    &= -e\int_{-\infty}^\infty  \cF^{c\mu} u_\mu \phi(q^\al(\tau))~d\tau
\end{align*}
hence finishing the proof. 
\end{proof}


\subsection{Particle interactions}

We now analyze $T^{cd} := T^{cd}(F_n, F_m)$.
\begin{lemma}\label{lem_part_int} Let the assumptions of Theorem \ref{thm_main} hold.
For $n\neq m$, we have $T^{cd}(F_n, F_m) \in \mathcal{D}'(\mathbb{M}^4)$. Moreover, the distributional divergence $\del_d T^{cd}$ is supported on both of the non-intersecting world-lines of $\fq_n,\fq_m\in \SWL$ and is given by
\begin{align}
- \del_d T^{cd}[\phi] = \half \left(e_m \int_{\fq_m} F_n^{cd} (\dot{q}_m)_d ~\phi~d\tau + e_n \int_{\fq_n} F_m^{cd} (\dot{q}_n)_d ~\phi~d\tau \right).
\end{align}
\end{lemma}
\begin{proof}
To ease notation we write $q_n = q, q_m = \widehat q, F = F_n, \widehat F = F_m$. Define the two world-tubes $W_\eps, \widehat W_\eps$ surrounding the two world-lines $q, \widehat q$ such that the world-tubes do not intersect in the compact support of the test function $\phi$. This is possible for all $\eps > 0$ sufficiently small.

Our task is to compute the distributional divergence of $T^{cd}(F, \widehat F)$. By inspecting the definition of $T^{cd}(F, \widehat F)$, we see that $T^{cd}(F, \widehat F)$ is smooth away from the world-lines $q, \widehat q$. To obtain information about the integrability close to the world-lines, we focus without loss of generality on $W_\eps$. The first step is to replace the terms containing $\wa F^{ab}, \wa \widehat F^{ab}$ by the Li\'enard-Wiechert fields according to Proposition \ref{prop_LW}. By using \ref{prop:as1}, we see that
\begin{align}
    \left.T^{cd}(F, \widehat F)\right|_{X^\al(\tau, r, \theta, \varphi)} = \calO(r^{-2})
\end{align}
where $X^\al(\tau, r, \theta, \varphi)$ defines retarded light-cone coordinates relative to the world-line $q$.
Hence 
$$
T^{cd}(F, \widehat F) \in L^1_{\text{loc}}(W_\eps).
$$
A similar argument holds on $\widehat{W}_\eps$ and hence $T^{cd}(F, \widehat F)$
defines a distribution in $\mathcal{D}'(\mathbb{M}^4)$. Now,
\begin{align*}
T^{ac}[\del_c \phi] &= \lim_{\eps\to 0^+} \int_{\mathbb{M}^4\setminus (W_{\eps}\cup \widehat{W}_{\eps})} T^{ac} \del_c\phi~d^4 x = \lim_{\eps\to 0^+}\left(- \int_{\del W_\eps} T^{ab}\phi~d\Sigma_b - \int_{\del \widehat{W}_\eps} T^{ab}\phi~d\Sigma_b
\right)
\end{align*}
where we have used the Gauss-Stokes theorem and $\del_c T^{ac} = 0$ outside of the world-lines. The same arguments as in the proof of Lemma \ref{lemma_distr_div} now lead to the result. 
\end{proof}


\subsection{Proof of Main Result (Theorem \ref{thm_main})}\label{sec::main_result_pf}

We can now finish the proof of our main result. Recall that the total field tensor is the sum of the particle fields $F^{ab} = \sum_{n=1}^N F_{n}^{a b}$ and hence
the energy-momentum tensor of the total field splits up into a sum of terms:
\begin{align*}
T_{field}= T(F, F) = \sum_{n=1}^N T(F_n) + \sum_{n,m;n\neq m} T(F_n, F_m).
\end{align*}
Now we apply Lemmas \ref{lemma_distr_div} and Lemma \ref{lem_part_int} to see that $T_{field}$ defines a distribution and that the divergence $\del_d T^{cd}$ satisfies
\begin{align*}
-\del_d T^{cd}_{field}[\phi] &= \sum_{n=1}^N e_n \int_{\fq_n} \cF^{c d}_n (\dot{q}_n)_d \;\phi ~d\tau \\
&\quad+ \half \sum_{n,m; n\neq m}\left(e_m\int_{\fq_m}F_n^{cd} (\dot{q}_m)_d \;\phi~d\tau + e_n\int_{\fq_n}F_m^{cd} (\dot{q}_n)_d ~\phi~d\tau \right)\\
&= \sum_{n=1}^N e_n \int_{\fq_n} \cF^{c d}_n (\dot{q}_n)_d \phi ~d\tau\\
&\quad+\half\sum_{m=1}^N \int_{\fq_m} \sum_{n;n\neq m} e_m F^{cd}_n (\dot{q}_m)_d \;\phi~d\tau +
\half\sum_{n=1}^N \int_{\fq_n} \sum_{m; m \neq n} e_n F^{cd}_m (\dot{q}_n)_d \;\phi~d\tau
\end{align*}
with $\cF_n$ being the self-field generated by the $n$-th particle. For any $\fq_p$, we can take $\phi$ to be a test function whose support is concentrated around the world-line $\fq_p$. Hence we obtain
\begin{align*}
-\del_d T^{cd}_{field}[\phi] &=
 e_p \int_{\fq_p} \cF^{c d}_p(\tau) (\dot{q}_p)_d \;\phi ~d\tau
+ \int_{\fq_p} \sum_{m=1, m\neq p}^N e_p F_m^{cd} (\dot{q}_p)_d \;\phi ~d\tau.
\end{align*}
Thus the equation of motion \eqref{eq_motion} is implied by \eqref{cons}, finishing the proof.

%
%

\section{Variational principle for retarded fields}\label{sec_var}


\subsection{Local formulation}

In this section, we construct a variational principle which implies the equations of motions \eqref{eq_motion} with the expression \eqref{eq_self1} for the self-force. For simplicity, we treat the case of a single particle with associated self-fields \eqref{eq_field_tensor} moving in an external field described by the external vector potential $A^a_{ext}(x^\alpha)$. Consider first the case that no particles are present.
Ordinarily, one requires that the action 
\begin{align}\label{eq_var_int1}
I := \int_{\R^4} \mathcal{L}(F)~d^4 x 
\end{align}
is stationary with respect to variations of the vector potential $A^\mu\mapsto A^\mu + \delta A^\mu$. The field equation \eqref{eq_field} is then a necessary consequence. The variation $\delta A^\mu$ is assumed to be smooth and of compact support. In contrast, our variational principle will be based on variations of the particle trajectory. The corresponding variation of the fields follows from the assumption that the fields are given by the retarded fields \eqref{eq_field_tensor2}.
As not enough information is available on the far-field behavior of the retarded fields, we will generalize the action principle by
\begin{enumerate}
    \item[(a)] replacing the integration over $\R^4$ by an integration over bounded domains $\Omega$ and
    \item[(b)] allow also variations of the vector potential that do not vanish on the boundary $\del \Omega$.
\end{enumerate}
As a consequence, the condition that the action functional is stationary is modified to include a boundary term. Similar variational principles also work for Einstein's equations \cite{Poisson2004}. 

As an intermediary step, we introduce a local version of the ordinary variational principle for BLTP fields without particles:
\begin{proposition}\label{stat_gen}
Let $A^c\in C^4(\R^4)$. For any bounded, open $\Omega\subset \R^4$ with Lipschitz boundary define the following local action functional $I_\Omega$ by
\begin{align}\label{eq_var11}
    I_{\Omega}[A^\mu] = \int_{\Omega} \mathcal{L}(F)~d^4 x
\end{align}
where $F_{cd} = \del_{[c} A_{d]}$ and
\begin{align}\label{eq_Lagrangian1}
\mathcal{L} = - \frac{1}{16\pi} F_{ab} F^{ab} - \frac{1}{8\pi \kappa^2} \del_c F^{ca} \del^d F_{da}.
\end{align}
Assume that for any $\Omega$ and any divergence-free variation $$\widehat{A}^c\in C^4(\bar{\Omega}), \del^d \widehat{A}_d = 0$$ the following holds:
\begin{align}\label{eq_stat_gen}
\left.\frac{d}{d\sigma} I[A^\mu+ \sigma \widehat{A}^\mu]\right|_{\sigma = 0} &= \int_{\del \Omega}~\left[\left(\frac{1}{4\pi \ka^2} \del^a \del_c F^{cb} - \frac{1}{4\pi} F^{ab}\right)\widehat{A}_b - \frac{1}{4\pi \ka^2} \del_c F^{cd} \del^a \widehat{A}_d\right] ~d\Sigma_a.
\end{align}
Then the field equation \eqref{eq_field} with $j^a = 0$ holds on $\R^4$.
\end{proposition}
\begin{proof}
Taking the derivative of $I_\Omega$,
\begin{align*}
    \left.\frac{d}{d\sigma}I_{\Omega}[A^\mu+\sigma \widehat{A}^\mu]\right|_{\sigma = 0} = \int_{\Omega} \left(-\frac{1}{8\pi} F^{ab} \del_{[a} \widehat{A}_{b]} - \frac{1}{4\pi \ka^2} \del_c F^{ca} \wa \widehat{A}_a\right)~d^4 x
\end{align*}
where we have used the Lorentz gauge $\del^d A_d = 0$. An integration by parts gives
\begin{align*}
    \left.\frac{d}{d\sigma}I_{\Omega}[A^\mu+\sigma \widehat{A}^\mu]\right|_{\sigma = 0} &= \int_{\Omega} \left(\frac{1}{4\pi} \del_a F^{ab} - \frac{1}{4\pi \ka^2} \wa \del_c F^{cb} \right)\widehat{A}_b~d^4 x\\
    &\quad+\int_{\del \Omega}~\left[\left(\frac{1}{4\pi \ka^2} \del^a \del_c F^{cb} - \frac{1}{4\pi} F^{ab}\right)\widehat{A}_b - \frac{1}{4\pi \ka^2} \del_c F^{cd} \del^a \widehat{A}_d\right] ~d\Sigma_a.
\end{align*}
Comparison with \eqref{eq_var11} gives 
$$
 \int_{\Omega} \left(\frac{1}{4\pi} \del_a F^{ab} - \frac{1}{4\pi \ka^2} \wa \del_c F^{cb} \right)\widehat{A}_b~d^4 x = 0
$$
for all $\widehat{A}\in C^4(\bar \Omega)$ and thus $(I - \ka^{-2}\wa) \del_c F^{c d} = 0$. The second equation of \eqref{eq_field} holds automatically since $\del_{[a} A_{b]} = F_{ab}$.
\end{proof}
\begin{remark}
The condition \eqref{eq_stat_gen} can be given the following interpretation: The field $F_{ab}$ is a \emph{generalized stationary point} of the functional $I_\Omega$ in the sense that the variations of the functional can be expressed by information on the boundary.  
\end{remark}


\subsection{Variation of retarded fields with particles}

For any space-time region $\Omega\subset \R^4$ with Lipschitz boundary, we consider the local action functional 
\begin{align}\label{functional}
\begin{split}
I_{\Omega}[\fq] &:= \int_{\Omega} \mathcal{L}(\del^{[c} A^{d]}[\fq])~d^4 x + e\int_{\fq\cap \Omega} A^{b}[\fq] ~d X_b +e \int_{\fq\cap \Omega} A^{b}_{ext}~d X_b - m_0 \int_{\fq\cap \Omega} ~d\tau
\end{split}
\end{align}
where $A_{ext}$ is an external vector potential and $A^{d}[\fq]$ is the vector potential associated to the subluminal world-line $\fq$, i.e.
\begin{align}\label{eq_pot2}
    A^d[\fq](x^\alpha) = e\ka \int_{-\infty}^{\tau_{ret}(x^\alpha)} \frac{J_1(\ka D(x^\alpha - q^\alpha(\bar \tau)))}{D(x^\alpha - q^\alpha(\bar \tau))} u^a(\bar \tau)~d \bar\tau.
\end{align}
The functional \eqref{functional} looks like the standard functional used to couple the field to a particle. However, we make variations of the trajectory only, and the variations of the potential follow from \eqref{eq_pot2}. This explains the usefulness of a generalized notion of the functional $I_\Omega$ being stationary, as variations of the particle trajectory will in general induce variations of the retarded field that do not vanish on $\del \Omega$. 

To prepare for the calculation of the variation of \eqref{functional}, we define the class of variations of a given world-line $\fq$ (cf. \cite{HawkingEllisBook}).

\begin{definition}\label{def_variation}
Let $\fq\in \SWL$ be a given subluminal world-line and let $\Omega \subset \R^4$ be a bounded Lipschitz domain. A variation of $\fq$, compactly supported in $\Omega$, is a family of world-lines defined by parametrizations
\begin{align}\label{eq_congrue}
Q^\alpha(\tau, \sigma)
\end{align}
where $Q^\alpha \in C^2(\R\times[-a, a])$ for some $a>0$, that satisfies the following:
\begin{enumerate}
\item[(i)] all world-lines are subluminal 
\item[(ii)] $\{Q^\al(\cdot, 0)\} = \fq$
\item[(iii)] There exists a bounded open set $\widetilde\Om\subset\subset\Omega$ such that $\{Q^\al(\cdot, \sigma)\} = \fq$ outside of $\widetilde\Omega$ for all $\sigma\in (-a, a)$. 
\end{enumerate}
In particular, the curves $\tau \mapsto Q^\alpha (\tau, \sigma)$ for $\sigma \neq 0$ are parameterized with respect to their own proper time and not with the proper time on $q^\al$.
For any $\sigma\in(-a,a)$, we denote the entry/exit time of $Q^\al$ into/out of $\Omega$ by $\tau_0(\sigma)$ and $\tau_1(\sigma)$ respectively and set $\tau_0=\tau_0(0)$, $\tau_1=\tau_1(0)$.
For a given variation of $\fq$ with a particular parameterization $Q^\alpha(\cdot, \cdot)$, we define the deviation vector field by
\begin{align}\label{def_devi}
\xi^\alpha(\tau, \sigma) := \partial_{\sigma} Q^\alpha(\tau, \sigma).
\end{align}
Note the following relation between $\xi^\al$ and $u^\al$:
\begin{align}
    \del_\tau\xi^\al&=\del_\sigma u^\al.
\end{align}
\end{definition}

\begin{remark} In this section we adapt notation from previous sections replacing $q^\al(\tau)$ by $Q^\al(\tau,\sigma)$. For example 
$R^\al=x^\al-Q^\al(\tau,\sigma)$, $u^\al(\tau,\sigma)=\del_\tau Q^\al(\tau,\sigma)$, $S=u_\ga R^\ga$, $X^\al_\sigma(\tau, r, \theta, \varphi) = Q^\al(\tau,\sigma) + r(u^\al(\tau,\sigma) + N^\al(\tau, \theta, \varphi,\sigma))$ etc. Note that all the familiar relations for these quantities hold, as well as the asymptotics from Proposition \ref{prop:as1} on the respective boundaries of world-tubes $W_{\eps,\sigma}$.
\end{remark}
For a given variation of $\fq$, define
\begin{align}\label{eq_pot3}
    A^a[Q(\tau, \sigma)](x^\alpha) := e\ka \int_{-\infty}^{\tau_{ret}(x^\alpha, \sigma)} \frac{J_1(\ka D(x^\al-Q^\al(\tau,\sigma))}{D(x^\al-Q^\al(\tau,\sigma))} ~\del_\tau Q^a~d\tau
\end{align}
i.e. $A^a[Q(\tau, \sigma)]$ is retarded vector potential associated to the world-line $Q(\tau, \sigma)$. 

\begin{lemma}\label{lem_devA1} Suppose $Q(\tau,\sigma)$ is variation of $\fq$. Then the variation of the potential $A^a[Q(\cdot, \sigma)]$ with respect to $\sigma$ for fixed $x^\alpha\in \R^4\setminus \{Q(\cdot, \sigma)\}$ is given by 
\begin{align}
\widehat{A}^a(x^\al,\sigma):=\del_\sigma A^a(x^\alpha,\sigma)  
    = G^a_0 + G^a_1 \quad (x^\alpha \in \R^4\setminus \{Q(\cdot, \sigma)\})
\end{align}
with
\begin{align}\label{eq_devA111}
     \begin{split}
    G_0^a(x^\alpha,\sigma) &:= -e\ka^2 \int_{-\infty}^{\tau_{ret}(x^\alpha, \sigma)} \frac{J_2(\ka D)}{D^2} \xi^b R_b u^a~d \tau
    + e\ka \int_{-\infty}^{\tau_{ret}(x^\alpha,\sigma)} \frac{J_1(\ka D)}{D} \del_\tau \xi^a~d \tau \\
    G_1^a(x^\alpha,\sigma) &:=- \left.\frac{e \ka^2}{2} \frac{u^a \xi_\gamma R^\gamma}{S}\right|_{\tau = \tau_{ret}(x^\alpha, \sigma)}
    \end{split}
\end{align}
where the integrals are absolutely convergent.
\end{lemma}

\begin{proof} For $\sigma \in (-a, a)$, the vector potential is defined by \eqref{eq_pot3},
where $\tau_{ret}(x^\alpha, \sigma)$ satisfies
\begin{align}
    g_{ab}(x^a - Q^a(\tau_{ret}(x^\alpha, \sigma), \sigma))(x^b - Q^b(\tau_{ret}(x^\alpha, \sigma), \sigma)) = 0. \label{eq_misc33}
\end{align}
Differentiation of \eqref{eq_misc33} with respect to $\sigma$ implies
\begin{align*}
\frac{\del \tau_{ret}}{\partial \sigma}(x^\alpha,\sigma)= - \left.\frac{\xi_\gamma R^\gamma}{S}\right|_{\tau = \tau_{ret}(x^\alpha, \sigma)}
\end{align*}
where $S = u^\gamma R_\gamma$. 
If $D = D(x^\alpha - Q^\alpha(\cdot, \sigma))$ then for fixed $x^\alpha, \tau$
$$
\frac{\del D}{\del \sigma} = \frac{\xi_\ga R^\ga}{D}.
$$
Using the foregoing relationships, the following formula for $\widehat A^a$ follows by direct computation (note that $\del_{\sigma} u_b = \del_{\tau} \xi_b$):
\begin{align*}
    \begin{split}
    \widehat{A}^a = 
    \del_\sigma A^a(x^\alpha, \sigma)
    &=
     -e\ka^2 \int_{-\infty}^{\tau_{ret}(x^\alpha,\sigma)} \frac{J_2(\ka D)}{D^2} \xi^b R_b u^a ~d\tau\\
    &\quad+ e\ka \int_{-\infty}^{\tau_{ret}(x^\alpha,\sigma)} \frac{J_1(\ka D)}{D} \del_\tau \xi^a~d\tau - \left.\frac{e \ka^2 u^a\xi_\gamma R^\gamma}{2 S}\right|_{\tau = \tau_{ret}(x^\alpha,\sigma)}.
    \end{split}
\end{align*}
The formal differentiation can be justified by similar arguments as those employed in the proof of Proposition \ref{prop_conv}. This proves \eqref{eq_devA111}.
\end{proof}

\begin{theorem}\label{thm_var} Let $A^\alpha_{ext}\in C^2(\R^4)$. Further assume that for all open, bounded Lipschitz domains $\Omega$ and all variations of $\fq$ in the sense of Definition \ref{def_variation},  
    \begin{align}
       \left.\frac{d}{d\sigma}I_\Omega\right|_{\sigma = 0} = \int_{\del \Omega}~\left[\left(\frac{1}{4\pi \ka^2} \del^a \del_c F^{cb} - \frac{1}{4\pi} F^{ab}\right)\del_\sigma A_b - \frac{1}{4\pi \ka^2} \del_c F^{cd} \del^a \del_\sigma A_b\right]_{\sigma=0} ~d\Sigma_a
    \end{align}
    holds.
Then $q^\alpha(\tau)$ satisfies the following equation of motion:
\begin{align}\label{eq_mot_var}
m_0 \ddot{q}^a = e\left( F_{ext}^{ab} +  \mathcal{F}^{ab}\right) \dot{q}_b
\end{align}
where $\cF^{ab}$ is the self-force defined by \eqref{eq_self1}.
\end{theorem}
To structure the proof, we first state some preparatory lemmas and formulas. Define 
\begin{align}\label{functionals}
\begin{split}
I_{1,\Omega}[\fq] &:= \int_{\Omega} \mathcal{L}(\del^{[c} A^{d]}[\fq])~d^4 x  \\
I_{2, \Omega}[\fq] &:= e \int_{\fq\cap \Omega} A^{b}[\fq]~d X_b\\
I_{3, \Omega}[\fq] &:= - m_0 \int_{\fq \cap \Omega} d\tau
+ e \int_{\fq\cap \Omega} A^{b}_{ext}~dX_b
\end{split}
\end{align}
so that $I_{\Omega} = I_{1,\Omega}+I_{2,\Omega}+I_{3,\Omega}$. Note that all integrals are well-defined. This can be checked using asymptotic expansions in Proposition \ref{prop:as1}.

\begin{lemma}\label{lem_var_I1}
The variation of $I_{1,\Omega}$ is given by 
\begin{align}
\begin{split}\label{eq_var_I1}
\left.\frac{d}{d\sigma} I_{1, \Omega}\right|_{\sigma = 0} 
&=  \int_{\del \Omega}~\left[\left(\frac{1}{4\pi \ka^2} \del^a \del_c F^{cb} - \frac{1}{4\pi} F^{ab}\right)\widehat{A}_b - \frac{1}{4\pi \ka^2} \del_c F^{cd} \del^a \widehat{A}_d\right]_{\sigma=0} ~d\Sigma_a\\
&\quad - e \int_{\fq\cap \Omega} \mathcal{G}^b_0 u_b ~d\tau +\int_{\fq\cap\Om} \frac{e^2\ka^2}{2}\xi^b u_b~d\tau
\end{split}
\end{align}
where
\begin{align}
\mathcal{G}_0^a(\tau) &= -\left.\mathcal{G}_0^a(\tau,\sigma)\right|_{\sigma=0} =  e\ka^2 \left.\int_{-\infty}^{\tau} \frac{J_2(\ka D)}{D^2} \xi^b R_b u^a~d \bar \tau\right|_{\sigma=0}
    + e\ka \left.\int_{-\infty}^{\tau} \frac{J_1(\ka D)}{D} \del_\tau \xi^a~d \bar\tau\right|_{\sigma=0}
\end{align}
and in this formula, $D = D(Q^\al(\tau, \sigma)-Q^\al(\bar \tau, \sigma)), R^\al = Q^\al(\tau, \sigma)-Q^\al(\bar \tau, \sigma)$.
\end{lemma}

\begin{proof} The proof of this Lemma involves the differentiation of a parameter-dependent integral. As the integrand is mildly singular at the world-line, which changes with variation in $\sigma$, some care is required in the following arguments. Let $X^\al_\sigma(\tau, r,\theta, \varphi)$ be the retarded light-cone coordinate mapping associated to the world-line $Q(\cdot, \sigma)$ (see \eqref{light_cone_coord}). First we note that by using the asymptotic expansion in Proposition \ref{prop:as1}, we obtain in a neigborhood of $\{Q(\cdot, \sigma)\}$
\begin{align}\label{eq:aux11}
    \mathcal{L}\circ X_\sigma(\tau, r, \theta, \varphi) = H_0 r^{-2} + H_1
\end{align}
with $H_0 = \frac{e^2 \ka^2}{8\pi}$ and the remainder $H_1(\tau, r, \theta, \varphi) = \calO(r^{-1})$ uniformly in $\sigma$ and $\tau,\theta, \varphi$ (as long as $\tau$ varies in a compact interval).
Moreover we compute 
\begin{align*}
    \frac{\del \mathcal{L}}{\del\sigma} = - \frac{1}{8\pi} F^{ab} \del_{[a} \widehat{A}_{b]} - \frac{1}{4\pi \ka^2} \del_c F^{ca} \wa \widehat{A}_a
\end{align*}
at all points away from the world-line $Q^{\al}(\tau,\sigma)$. Here we also used Proposition \ref{prop_A_div} to conclude that $\del^d \widehat{A}_d = 0$.
By again using the asymptotics provided in Proposition \ref{prop:as1}, Lemma \ref{lem_devA1} and the relation $\wa \widehat{A}_a = -\ka^2(\del_\sigma U_a - \widehat{A}_a)$ we see that
\begin{align}\label{eq:aux12}
    \frac{\del \mathcal{L}}{\del \sigma}\circ X_\sigma(\tau, r, \theta, \varphi) = \calO(r^{-2}).
\end{align}
Next, let $\chi$ be a smooth cut-off function with support as depicted in Figure \ref{figChi}.
\begin{figure}[h]
    \centering
    \includegraphics[scale=0.4]{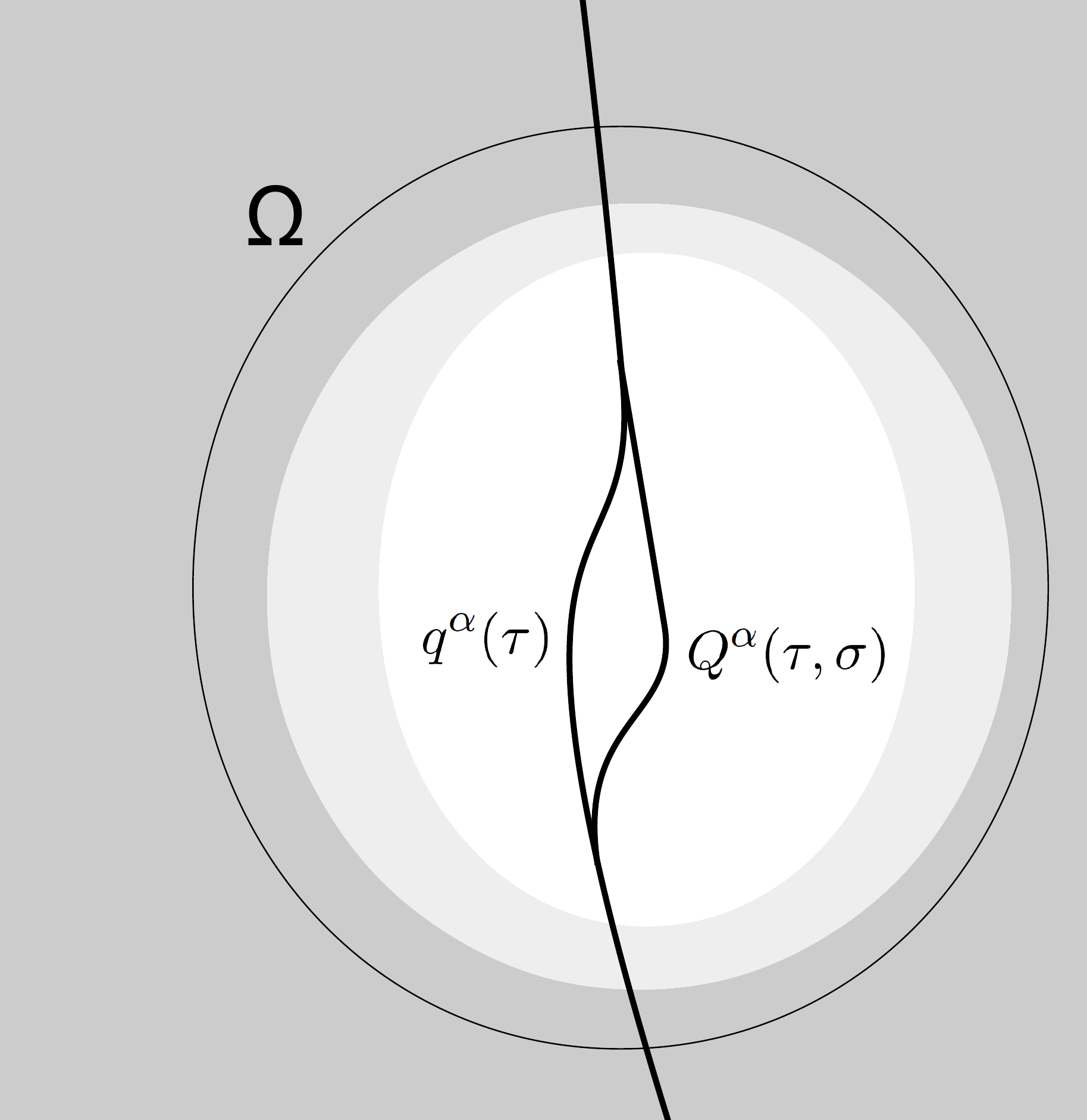}
    \caption{A world-line $q^{\al}(\tau)$ and a variation $Q^\al(\tau,\sigma)$ in $\Om$. White area: $\chi=1$, light grey area: $0<\chi< 1$, dark grey area: $\chi=0$.}
    \label{figChi}
\end{figure}

\noindent
Note that the curves $Q(\cdot, \sigma)$ coincide with $Q(\cdot, 0)$ outside of $\{\chi=1\}$. The integral $I_{1,\Omega}$ can be split into 
\begin{align*}
    I_{1,\Omega} = \int_{\Omega} (1-\chi)\mathcal{L} + \int_{\R^4} \chi \mathcal{L}.
\end{align*}
In view of \eqref{eq:aux11} and \eqref{eq:aux12}, standard results for differentiation of parameter-dependent integrals guarantee that
\begin{align}
    \del_\sigma\int_{\Omega} (1-\chi)\mathcal{L} = \int_{\Omega} (1-\chi)\frac{\del \mathcal{L}}{\del \sigma},
\end{align}
the key point being that there exists a integrable majorant function $\Phi\in L^1(\R^4)$ with $\left|\frac{\del \mathcal{L}}{\del \sigma}\right|\leq \Phi$ on $\Omega\setminus \{\chi=1\}$, independent of $\sigma$. 

Now define $H(x^\al) = \chi\mathcal{L}(\del^{[c} A^{d]})$ and apply Lemma \ref{lem_dev_param} from Appendix D. For this, observe in particular that in view of \eqref{eq:aux11}, $H\circ X_\sigma$ can be written in the form $H_0 r^{-1} + H_1$ with the properties required for Lemma \ref{lem_dev_param}.
Therefore 
\begin{align*}
    \del_\sigma \int_{\R^4} \chi \mathcal{L}~d x  = \int_{\R^4} \chi \frac{\del\mathcal{L}}{\del\sigma}~d x = 
    \int_{\Omega} \chi \frac{\del\mathcal{L}}{\del\sigma}~d x
\end{align*}
with an absolutely convergent integral. Hence we have
\begin{align*}
    \del_\sigma I_{1,\Omega} = \int_\Omega \frac{\del\mathcal{L}}{\del\sigma}~d x.
\end{align*}
We now determine the value of this integral by cutting out a world-tube $W_{\eps,\sigma}$ and applying integration by parts. We have 
\begin{align*}
   \frac{d}{d\sigma} I_{1,\Omega} = \lim_{\eps\to 0^+}\int_{\Omega\setminus W_{\eps,\sigma}} - \frac{1}{8\pi} F^{ab} \del_{[a} \widehat{A}_{b]} - \frac{1}{4\pi \ka^2} \del_c F^{ca} \wa \widehat{A}_a ~d^4 x =: \lim_{\eps\to 0^+} J_\eps.
\end{align*}
$J_\eps$ can be written as
\begin{align*}
    J_\eps &= \int_{(\del \Omega)\setminus W_{\eps,\sigma}}~\left[\left(\frac{1}{4\pi \ka^2} \del^a \del_c F^{cb} - \frac{1}{4\pi} F^{ab}\right)\widehat{A}_b - \frac{1}{4\pi \ka^2} \del_c F^{cd} \del^a \widehat{A}_d\right] ~d\Sigma_a \\
    &\quad- \int_{(\del W_{\eps,\sigma}) \cap \Omega}~\left[\left(\frac{1}{4\pi \ka^2} \del^a \del_c F^{cb} - \frac{1}{4\pi} F^{ab}\right)\widehat{A}_b - \frac{1}{4\pi \ka^2} \del_c F^{cd} \del^a \widehat{A}_d\right] ~d\Sigma_a\\
    &=: J_{1,\eps} - J_{2, \eps}
\end{align*}
where the orientation of $d\Sigma_a$ is chosen as outward-oriented on both $\del W_{\eps,\sigma}$ and $\del \Omega$. We now wish to send $\eps \to 0^+$ for a fixed $\sigma$ and since $Q(\tau, \sigma)$ is parameterized in such a way that $\tau$ is the proper time on $Q(\cdot, \sigma)$, we can use the asymptotics on $\del W_{\eps,\sigma}$ provided in Proposition \ref{prop:as1} to write
\begin{align*}
    &G_0^a(X^\al_\sigma(\tau, \eps, \theta, \varphi)) = \mathcal{G}_0^a(\tau,\sigma) + \mathcal{O}(\eps)\\
    &G_1^a(X^\al_\sigma(\tau, \eps, \theta, \varphi)) = \frac{e \ka^2}{2} u^a \xi \cdot (u + N)\\
    &\del_c G_0^a(X^\al_\sigma(\tau, \eps, \theta, \varphi)) = \mathcal{O}(1) \\
    &\del_c G_1^a(X^\al_\sigma(\tau, \eps, \theta, \varphi)) = -\frac{e\ka^2 u^a}{2\eps} \left(-\xi_c 
    - u_c (\xi\cdot u) + (\xi\cdot N) N_c\right) + \mathcal{O}(1)\\
    & \del_c F^{cd} = -\frac{e\ka^2 u^d}{\eps}+\mathcal{O}(1)\\
    & \del_a \del_c F^{cb} = \frac{e \ka^2 u^b N_a}{\eps^2} + \mathcal{O}(\eps^{-1})\\
    &\om_d=N_d+\mathcal{O}(\eps)
\end{align*}
where the $\mathcal{O}$ terms are uniform in $\sigma,\tau,\theta, \varphi$ for $\tau$ in a compact interval.
We write $\widehat{A}^d = G_0^d + G_1^d$ and write $J_{2, \eps}$ as a sum of two terms $J_{20, \eps}+J_{21, \eps}$: 
\begin{align*}
    J_{2i, \eps} = \int_{(\del W_{\eps,\sigma})\cap \Omega}~\left[\left(\frac{1}{4\pi \ka^2} \del^a \del_c F^{cb} - \frac{1}{4\pi} F^{ab}\right)(G_i)_b - \frac{1}{4\pi \ka^2} \del_c F^{cd} \del^a (G_i)_d \right] ~d\Sigma_a
\end{align*}
where $i=0, 1$. Consider first $J_{21,\eps}$. Upon using the asymptotic expansions, we get
\begin{align*}
     \left.\lim_{\eps\to 0^+} J_{21, \eps}\right|_{\sigma=0} = -\left.\int_{\tau_0(\sigma)}^{\tau_1(\sigma)} \frac{e^2\ka^2}{2}\xi^b u_b~d\tau\right|_{\sigma=0}.
\end{align*}
As for $J_{20,\eps}$, using the above asymptotics we obtain
\begin{align*}
    \lim_{\eps\to 0^+} \left. J_{20,\eps}\right|_{\sigma=0} &= e \int_{\fq\cap \Omega}  \mathcal{G}^b_0 u_b~d\tau.
\end{align*}
The part $J_{1\eps}$ gives in the limit $\eps \to 0^+$
\begin{align*}
     \lim_{\eps\to 0^+} J_{1\eps}=\int_{\del \Omega}~\left[\left(\frac{1}{4\pi \ka^2} \del^a \del_c F^{cb} - \frac{1}{4\pi} F^{ab}\right)\widehat{A}_b - \frac{1}{4\pi \ka^2} \del_c F^{cd} \del^a \widehat{A}_d\right] ~d\Sigma_a.
\end{align*}
\end{proof}

\begin{lemma}
The variation of $I_{2,\Omega}$ is given by 
\begin{align}\label{var_I2}
    \left.\frac{d}{d\sigma}I_{2, \Omega}[Q(\cdot, \sigma)]\right|_{\sigma = 0} &= e\int_{\fq\cap \Omega} \mathcal{G}^b_0 u_b ~d\tau +
    e\int_{\fq\cap \Omega} \cF^{cb}u_c \xi_b ~ d\tau-\int_{\fq\cap\Om}\frac{e^2\ka^2}{2}\xi^b u_b~d\tau
\end{align}
where all quantities on the right hand side are evaluated at $\sigma=0$.
\end{lemma}
\begin{proof} In order to calculate the variation of $I_{2, \Omega}$ we need to vary the vector potential $A^a$ on the world-line with respect to $\sigma$, where the vector potential is generated by $Q^\alpha(\tau,\sigma)$. This differs from $\widehat{A}^a$, which was the derivative of $A^a$ with respect to $\sigma$ for a fixed $x^\alpha$ off the world-line. To shorten notation, we define $\mathcal{A}^a(\tau, \sigma)$ by
$$
\mathcal{A}^a(\tau, \sigma) := e \ka \int_{-\infty}^\tau \frac{J_1(\ka D(Q^\alpha(\tau, \sigma)-Q^\alpha(\bar\tau, \sigma)))}{
D(Q^\alpha(\tau, \sigma)-Q^\alpha(\bar\tau, \sigma))}u^a(\bar \tau, \sigma)~d\bar \tau = A^a[Q(\cdot, \sigma)](Q^\alpha(\tau, \sigma))
$$
and after observing
\begin{align*}
    \left.\del_\sigma \frac{J_1(\ka D(Q^\alpha(\tau, \sigma)-Q^\alpha(\bar\tau, \sigma)))}{
D(Q^\alpha(\tau, \sigma)-Q^\alpha(\bar\tau, \sigma))}\right|_{\sigma = 0} 
&=\ka \frac{J_2(\ka D)}{D^2} (\xi^b(\tau)-\xi^b(\bar \tau)) (q_b(\tau)- q_b(\bar \tau))
\end{align*}
we compute
\begin{align*}
    \left.\del_{\sigma} \mathcal{A}^a(\tau, \sigma)\right|_{\sigma = 0} &= -e\ka^2 \int_{-\infty}^\tau \frac{J_2(\ka D)}{D^2} \xi^b(\bar \tau) (q_b(\tau)-q_b(\bar \tau)) u^a(\bar \tau)~d\bar \tau  + e\ka \int_{-\infty}^{\tau} \frac{J_1(\ka D)}{D} \del_\tau \xi(\bar \tau)~d\bar \tau \\
    &\quad+ e\ka^2\int_{-\infty}^\tau 
    \frac{J_2(\ka D)}{D^2} \xi^b(\tau) (q_b(\tau)-q_b(\bar \tau)) u^a(\bar \tau)~d\bar \tau.
\end{align*}
Notice that this can be written as $G^a_0(q^\alpha(\tau)) + \mathcal{K}^a(\tau)$, or $\mathcal{G}^a_0(\tau) + \mathcal{K}^a(\tau)$, where $\mathcal{K}^a$ is defined by
\begin{align*}
    \mathcal{K}^a(\tau) :=  e\ka^2\int_{-\infty}^\tau 
    \frac{J_2(\ka D)}{D^2} \xi^b(\tau) (q_b(\tau)-q_b(\bar \tau)) u^a(\bar \tau)~d\bar \tau.
\end{align*}
Moreover,
\begin{align*}
    \left.\del_{\tau} \mathcal{A}^a(\tau, \sigma)\right|_{\sigma = 0} &= \frac{e \ka^2}{2} u^a(\tau) 
     + e \ka^2 \int_{-\infty}^\tau \frac{J_2(\ka D)}{D^2}u^b(\tau) (q_b(\tau) - q_b(\bar\tau))u^a(\bar \tau)~d\bar \tau=: \frac{e \ka^2}{2} u^a(\tau)+ \widetilde{\mathcal{K}}^a.
\end{align*}
A calculation yields for $\sigma=0$
\begin{align}\label{eq_key}
\left.\mathcal{K}^b u_b - \widetilde{\mathcal{K}}^b \xi_b\right. = \cF^{bc}\xi_b u_c.
\end{align}
Now we compute the variation of $I_{2,\Omega}$:
\begin{align*}
     &\frac{1}{e}\frac{d}{d\sigma} I_{2, \Omega} =\frac{d}{d\sigma} \int_{\tau_0(\sigma)}^{\tau_{1}(\sigma)} \mathcal{A}^{b} u_b~d\tau \\
     &= \int_{\tau_0(\sigma)}^{\tau_{1}(\sigma)} \frac{d}{d\sigma} \left(\mathcal{A}^{b} u_b\right)~d\tau +\mathcal{A}^{b}(\tau_1(\sigma), \sigma)u_b(\tau_1(\sigma), \sigma)\del_\sigma\tau_1(\sigma)-\mathcal{A}^{b}(\tau_0(\sigma), \sigma)u_b(\tau_0(\sigma), \sigma)\del_\sigma\tau_0(\sigma) \\
     &= \int_{\tau_0(\sigma)}^{\tau_{1}(\sigma)} \del_\sigma\mathcal{A}^{b} u_b+\mathcal{A}^{b}\del_\sigma u_b~d\tau +\mathcal{A}^{b}(\tau_1(\sigma), \sigma)u_b(\tau_1(\sigma), \sigma)\del_\sigma\tau_1(\sigma)-\mathcal{A}^{b}(\tau_0(\sigma), \sigma)u_b(\tau_0(\sigma), \sigma)\del_\sigma\tau_0(\sigma)\\
     &\quad+\mathcal{A}^{b}(\tau_1(\sigma), \sigma)u_b(\tau_1(\sigma), \sigma)\del_\sigma\tau_1(\sigma)-\mathcal{A}^{b}(\tau_0(\sigma), \sigma)u_b(\tau_0(\sigma), \sigma)\del_\sigma\tau_0(\sigma)
\end{align*}
where in the last step we used $\del_\sigma u_b = \del_\tau\xi_b$ and performed an integration by parts. Now for $i=0,1$ we have
\begin{align*}
    \xi^b(\tau_i(\sigma), \sigma)+u^b(\tau_i(\sigma), \sigma)\del_\sigma\tau_i(\sigma)
    &= (\del_\sigma Q^b)(\tau_i(\sigma), \sigma) +(\del_\tau Q^b)(\tau_i(\sigma), \sigma)\del_\sigma\tau_i(\sigma)\\
    &=\del_\sigma\left(Q^b(\tau_i(\sigma),\sigma)\right)=0
\end{align*}
by Assumption (iii) of Definition \ref{def_variation} (note that the coordinate of the points where the $Q(\cdot,\sigma)$ enter and exit $\Omega$ do not change with $\sigma$). This means that the boundary terms above sum up to 0. We have now:  
\begin{align*}
     \left.\frac{1}{e}\frac{d}{d\sigma} I_{2, \Omega}\right|_{\sigma=0} &=\left.\int_{\tau_0(\sigma)}^{\tau_{1}(\sigma)} \del_\sigma\mathcal{A}^{b} u_b-\del_\tau\mathcal{A}^{b} \xi_b~d\tau\right|_{\sigma=0}= \int_{\tau_0}^{\tau_{1}} \left[\mathcal{G}_0^b u_b + \mathcal{K}^b u_b -\left( \frac{e \ka^2}{2} u^b + \widetilde{\mathcal{K}}^b\right)\xi_b \right]~d\tau\\
    &= \int_{\tau_0}^{\tau_{1}} \mathcal{G}_0^b u_b~d\tau +
    \int_{\tau_0}^{\tau_{1}} \left[\mathcal{K}^b u_b - \widetilde{\mathcal{K}}^b \xi_b\right] ~d\tau - \int_{\tau_0}^{\tau_1} \frac{e\ka^2}{2} \xi^b u_b~d\tau\\
    &= \int_{\tau_0}^{\tau_{1}} \mathcal{G}_0^b u_b ~d\tau +
    \int_{\tau_0}^{\tau_{1}} \cF^{bc}u_c \xi_b ~ d\tau - \int_{\tau_0}^{\tau_1} \frac{e\ka^2}{2} \xi^b u_b~d\tau
\end{align*}
where the quantities on the right hand side are evaluated at $\sigma = 0$.  We have used \eqref{eq_key} in the final step. 
\end{proof}

\begin{proof}[Proof of Theorem \ref{thm_var}]
The computation of the variation of $I_{3,\Omega}$ is standard and gives
\begin{align}
    \left.\frac{d}{d\sigma} I_{3,\Omega}\right|_{\sigma = 0} = \int_{\fq} \left(-m_0 \frac{d u_b}{d\tau} + (F_{ext})_{ab} u^a \right) \xi^b ~d\tau 
\end{align}
where $(F_{ext})_{ab} = \del_{[a} (A_{ext})_{b]}$. Now combining with \eqref{eq_var_I1} and \eqref{var_I2}, we obtain 
\begin{align}\label{eq_aux2}
\begin{split}
    \frac{d}{d\sigma} \left.I_\Omega\right|_{\sigma = 0} &= B
    +\int_{\fq} e\mathcal{F}_{bc} u^c \xi^b ~d\tau+ \int_{\fq} \left(-m_0 \frac{d u_b}{d\tau} + (F_{ext})_{ab} u^a \right) \xi^b ~d\tau
\end{split}
\end{align}
where $B$ denotes the boundary term in \eqref{eq_var_I1} (integral over $\del\Omega$). Recall that the assumption of Theorem \ref{thm_var} was that $\frac{d}{d\sigma} \left.I_\Omega\right|_{\sigma = 0} = B$ for all variations of $\fq$. We deduce that 
\begin{align*}
m_0 \frac{d u_b}{d\tau} = e\mathcal{F}_{bc} u^c + (F_{ext})_{bc} u^c
\end{align*}
i.e. the equation of motion \eqref{eq_mot_var} holds.
\end{proof}

%
%

\section{Motion under influence of external forces}

In this section, we present a global existence result for the equations of motion of a single particle with a given external force. 
In the following, we fix a choice of Lorentz frame and work in $3+1$-dimensional notation, i.e. $x^a=(t, x)$, $x=(x^1,x^2,x^3)$. The self-fields take the form 
\begin{align}\label{eq_self}
\begin{split}
E(t) &=- e\ka^2 \int_{-\infty}^t \frac{J_2(\ka D)}{D^2}[q(t)-q(t')- v(t') (t-t')]~dt'\\
B(t) &= -e\ka^2 \int_{-\infty}^t \frac{J_2(\ka D)}{D^2}v(t')\times (q(t)-q(t'))~dt'
\end{split}
\end{align}
and the self-force three-vector is given by $e(E+ v\times B)$. Thus the equation of motion \eqref{eq_motion} reads
\begin{align}\label{eq_mot}
m_0 \frac{d}{dt}\left(\frac{v}{\sqrt{1-v^2}}\right) = e( E(t) + v(t) \times B(t) ) + F_{ext}(t)
\end{align}
where the self-fields are given by \eqref{eq_self} and where we have added an external, time-dependent force $F_{ext}$. The history-dependence of the self-force prevents us from considering an initial value problem for \eqref{eq_mot}. We therefore impose the condition that $v$ is asymptotically constant in the infinite past. By adjusting the Lorentz frame, we may assume that the particle was asymptotically at rest, i.e.
\begin{align}\label{cond_init}
\lim_{t\to-\infty}v(t) = 0.
\end{align}
\begin{theorem}\label{thm_exist}
Suppose that the external force $F_{ext}:\R \to \R^3$ is continuous and $F_{ext} = 0$ on 
$(-\infty, A_0]$ for some $A_0$. Then \eqref{eq_mot} has a solution $v \in C^1(\R)$ satisfying \eqref{cond_init} and such that
\begin{align}
\sup_{t\in (-\infty, A]} |v(t)| < 1
\end{align}
for all $A\in \R$. Moreover, we have 
\begin{align}\label{zero_left}
v(t) = 0 \quad (t \leq A_0)
\end{align}
and $v$ is unique among all solutions satisfying \eqref{zero_left}.
\end{theorem}

To start, we first reformulate the equation of motion \eqref{eq_mot} as an integral equation. The values of $e, \ka, m_0$ will not matter in the following and we set all of them equal to $1$. The relativistic momentum is  
$$
p = \frac{v}{\sqrt{1-v^2}}
$$
and the velocity can be recovered from
\begin{align}\label{eq_vel_p}
v = p/\sqrt{1+p^2}.
\end{align}
Integrating \eqref{eq_mot}, we arrive at
\begin{align}\label{eq_int1}
p(s) = \int_{-\infty}^s (E + p/\sqrt{1+p^2} \times B)~dt + \int_{-\infty}^s F_{ext}~dt
\end{align}
where the self-fields are defined by 
\begin{align}
\begin{split}
E(t) = E[p](t) &= \int_{-\infty}^t G(t-t', q(t)-q(t')) \int_{t'}^t [v(s)-v(t')]~ds ~dt'\\
B(t) = B[p](t) &= \int_{-\infty}^t G(t-t', q(t)-q(t')) v(t')\times \int_{t'}^t v(s)~ds ~dt'
\end{split}
\end{align}
and $v$ can be expressed by $p$ via the relation \eqref{eq_vel_p}. For given $p$, the expression $q(t)-q(t')$ is defined by 
\begin{align}\label{def_qt}
q(t)-q(t') = \int_{t'}^t p/\sqrt{1+p^2}~ds.
\end{align}
The kernel $G$ is given by
\begin{align}\label{def_G}
G(t-t', q(t)-q(t')) = -\frac{J_2(\sqrt{|t-t'|^2-|q(t)-q(t')|^2})}{|t-t'|^2-|q(t)-q(t')|^2}.
\end{align}

Let the spaces $\calC(-\infty, A], \calC(-\infty, A)$ be defined by
\begin{align}
\begin{split}
\calC(-\infty, A] &= \{ p : (-\infty, A] \to \R^3 : \text{$p$ is continuous, bounded and}\lim_{t\to -\infty} p(t) = 0\}\\
\calC(-\infty, A) &= \{ p : (-\infty, A) \to \R^3 : p\in \calC(-\infty,A-\eps]~\text{for all}~\eps>0\}.
\end{split}
\end{align}
$\calC(-\infty, A]$ is equipped with the usual $\|\cdot\|_{\infty}$ norm. 

\begin{lemma}\label{lem_cont}
\begin{enumerate}
\item[(a)]
Suppose that  $p\in \calC(-\infty, A]$ is a solution of the integral equation \eqref{eq_int1}. Then there exists a $\delta > 0$ and a unique solution $\tilde{p}\in \calC(\infty, A+\delta]$
of \eqref{eq_int1} such that $\tilde{p}(t) = p(t)$ for $t\leq A$. 
\item[(b)] Suppose now that $p \in \calC(-\infty, A)$ solves the integral equation \eqref{eq_int1} on $(-\infty, A)$ and that moreover
\begin{align}\label{eq_cond_cont}
\sup_{t\in (-\infty, A)} |E[p](t)|+|B[p](t)| < \infty
\end{align}
holds. Then the solution $p$ can be continued beyond $t = A$.
\end{enumerate}
\end{lemma}
\begin{proof}
\emph{(a)} Let a solution $p \in \calC(-\infty, A]$ be given. Define the space
$$
X_{\delta, M} := \{ \bar p\in \calC[A, A+\delta]: \bar p(A) = p(A),~|\bar p(s) -  p(A)| \leq M \}
$$
for given $\delta, M > 0$. $X_{\delta, M}$ is a complete metric space with distance given by $d(\bar p_1, \bar p_2) = \|\bar p_1 - \bar p_2\|_{\infty,[A,A+\delta]}$.
Our goal is now to define a suitable fixed-point operator. 
For any given $\bar p\in X_{\delta, M}$ extend the function $\bar p$ for $t\leq A$ by
\begin{align}\label{eq_ext}
\widetilde{p}(s) = \left\{ \begin{array}{ll}
p(s) &: s \leq A\\
\bar p(s) &: s \in [A, A+\delta]
\end{array}
\right.
\end{align}
and set
\begin{align}
\calK[\bar p](s) = p(A) + \int_A^s \left( E[\tilde p](t) + \frac{\bar p}{\sqrt{1+{\bar p}^2}}\times B[\tilde p](t)\right)~dt + \int_A^s F_{ext}(t)~dt.
\end{align}
The fixed-point equation to consider is:
\begin{align}\label{eq_int2}
\bar p(s) = \calK[\bar p](s)\quad (s\in [A, A+\delta]).
\end{align}
It is straightforward to check that for any solution $\bar p$ of \eqref{eq_int2} the function defined by \eqref{eq_ext}
solves \eqref{eq_int1} on $(-\infty, A+\delta]$.

One can show that 
$|\calK[\bar p]-p(A)|\leq C_1(M)\delta$ with an explicitly computable constant $C=C_1(M)$. If $\delta>0$ is chosen small enough, the operator $\calK$ maps $X_{\delta, M}$ into itself. Next, we observe that 
\begin{align}
\|\calK[\bar p_1] -\calK[\bar p_2]\|_{\infty}\leq C_2(M, \delta)\|\bar p_1-\bar p_2\|_{\infty}.
\end{align}
This follows from somewhat tedious computations and estimations using Proposition \ref{prop_estD}. $C_2$ can be computed explicitly and $C_2 < 1$ for small enough $\delta > 0$. Hence, a standard application of the contraction mapping theorem finishes this part of the proof.

\emph{(b)} The condition \eqref{eq_cond_cont} and the the fact that $p$ solves \eqref{eq_int1} imply that $p \in \calC(-\infty, A]$. Thus, we are in the situation of part \emph{(a)}.
\end{proof}

\begin{proposition}\label{prop_estD}
Let $p\in \calC(-\infty, A]$ be a solution as in Lemma \ref{lem_cont}. Let $\bar p_1, \bar p_2\in X_{\delta, M}$ and let $\tilde{p}_1, \tilde{p}_2$ be the corresponding extensions to $(-\infty, A+\delta]$ as in \eqref{eq_ext}. Denote
$$
D_i(t,t'): = \sqrt{|t-t'|^2 - |q_i(t)-q_i(t')|^2}
$$
and $v_i$ is defined from $\tilde{p}_i$ by \eqref{eq_vel_p} and $q_i(t)-q_i(t')$ by \eqref{def_qt}.
We have the the following estimates for all $ t'\leq t \in (-\infty, A+\delta]$: 
\begin{align}
&D_i^2(t,t') \geq K_1 |t-t'|^2\label{est1}\\
& \frac{|J_2(D_i)|}{D_i^2} \leq \frac{K_2}{(1+D^2_i)^{5/4}}\label{est4}\\
&|D_1-D_2|\leq K_3\delta \|\bar p_2-\bar p_1\|_{\infty}\label{est2}\\
\begin{split}
&|G(t-t', q_1(t)-q_1(t')) - G(t-t', q_2(t)-q_2(t'))|\leq K_4\delta\frac{\|\bar p_2-\bar p_1\|_{\infty}}{(1+|t-t'|^{2})^{5/4}}
\end{split}\label{est3}
\end{align}
with constants $K_1, K_2, K_3, K_4$ depending on $M$ and $\|p\|_\infty$ ($p$ being the solution on $(-\infty, A]$), but independent of $\delta$, if $\delta$ is sufficiently small.
\end{proposition}
\begin{proof}
To show \eqref{est1}, note that the function $x\mapsto \frac{x}{\sqrt{1+x^2}}$ is monotone for $x \geq 0$. We have
\begin{align}\label{est_p}
|q_i(t)-q_i(t')|\leq \int_{t'}^t \frac{|\tilde p_i|}{\sqrt{1+|\tilde p_i|^2}}~ds \leq 
|t-t'| \frac{\|p\|_\infty+M}{\sqrt{1+(\|p\|_\infty+M)^2}}
\end{align}
since $\|\tilde p_i\|\leq \|p\|_\infty + M$.
So the constant $K_1$ can be taken as
$$
K_1 = 1- \frac{\|p\|_\infty+M}{\sqrt{1+(\|p\|_\infty+M)^2}}
.
$$
\eqref{est4} is a direct consequence of \eqref{est1} and \eqref{eq_Kiess}.
We consider \eqref{est2}. Note first the following inequality for the $i$-th component of $x,y\in \R^3$
\begin{align}\label{basic_est}
\left|\frac{x^i}{\sqrt{1+x^2}} - \frac{y^i}{\sqrt{1+y^2}}\right|\leq K|x - y|\qquad i=1,2,3
\end{align}
holding with a universal constant $K>0$. This simply follows from the fact that the function $f_j(x)= x^j/\sqrt{1+x^2}$ has a globally bounded gradient. Next we observe that
\begin{align}
D_1 - D_2 &=\frac{(|q_2(t)-q_2(t')|- |q_1(t)-q_1(t')|)\left(|q_2(t)-q_2(t')|+|q_1(t)-q_1(t')|\right)}{D_1 + D_2}
\end{align}
and thus by using \eqref{est1} to estimate the denominator and the reverse triangle inequality in the numerator as well as \eqref{est_p},
\begin{align*}
|D_1(t,t')-D_2(t,t')|\leq C|q_2(t)-q_2(t')-(q_1(t)-q_1(t'))|\leq C\int_{t'}^t |v_2-v_1|~ds.
\end{align*}
Inserting \eqref{eq_vel_p} and using \eqref{basic_est} we arrive at
\begin{align*}
|D_1(t, t')-D_2(t, t')|\leq C \int_{t'}^t | \tilde p_1 - \tilde p_2 |~ds \leq C\delta \|\bar p_1 - \bar p_2\|_\infty
\end{align*}
by observing that $\tilde p_1-\tilde p_2 = 0$ for $s\leq A$. This leads finally to \eqref{est2}. To prove \eqref{est3}, we apply Lemma \ref{lem_besseldifference} to obtain
\begin{align*}
\left|\frac{J_2( D_1)}{D_1^2} - \frac{J_2(D_2)}{D_2^2}\right| 
&\leq C\frac{K_3\delta\|\bar p_2-\bar p_1\|_{\infty}}{(1+K_1|t-t'|^{2})^{5/4}}=K_4\delta\frac{\|\bar p_2-\bar p_1\|_{\infty}}{(1+|t-t'|^{2})^{5/4}}.
\end{align*}
\end{proof}

\begin{proof}[Proof of Theorem \ref{thm_exist}]
We want to prove the existence of a $p$ solving \eqref{eq_int1} that is defined on $\R$. Since $F_{ext} = 0$ on $(-\infty, A_0]$, a direct calculation shows that $p(t) = 0$
defines a solution on $(-\infty, A_0]$. This solution is can be extended past $t= A_0$ by Lemma \ref{lem_cont}. Let $\tilde{p}$ be the unique, maximally extended solution of \eqref{eq_int1}, defined on $(-\infty, A_1)$.

We claim that $A_1 = \infty$. To prove this, we assume that $A_1 <\infty$ and obtain a contradiction. From \eqref{eq_self} and \eqref{eq_Kiess} the electric self-field admits the estimate
\begin{align}\label{eq_estimate1}
|E[\tilde p](t)|\leq  \int_{-\infty}^t \frac{2|t-t'|}{\left[1+D^2(t, t')\right]^{5/4}}~dt'.
\end{align}
Since $v(t') = 0$ for $t' < A_0$, we have for $t' \leq A_0 \leq t$ the estimate
$$
|q(t)-q(t')| \leq \int_{A_0}^t |v(s)|~ds \leq (t-A_0)  
$$
holds and hence $D^2(t,t')\geq |t-t'|^2\left(1-\left|\frac{t-A_0}{t-t'}\right|^2\right)$ for $t' < A_0 < t$. Hence the integral in \eqref{eq_estimate1} is finite, since it can be split up into an integral from $-\infty$ to $A_0-1$ and a part from $A_0-1$ up to $t$ and for $-\infty < t' < A_0-1$, the integrand can be dominated by 
$$
\frac{|t-t'|}{\left[1+|t-t'|^2\left(1-\left|\frac{t-A_0}{t-t'}\right|^2\right)\right]^{5/4}}.
$$ 
Straightforward estimations show that
\begin{align}\label{eq_estimate2}
|E[\tilde p](t)|\leq  \int_{-\infty}^{A_0-1} \frac{2|t-t'|}{\left[1+|t-t'|^2\left(1-\left|\frac{t-A_0}{t-t'}\right|^2\right)\right]^{5/4}}
~dt' +  |t-A_0+1|^2
\end{align}
which has an upper bound independent of $t\in (A_0, A_1)$. Hence $\sup_{t\in (A_0, A_1)}|E[\tilde{p}]| < \infty$. 
A similar argument holds for $B[\tilde p]$, and we see that 
$$
\sup_{t\in (-\infty, A_1)} |E[\tilde p](t)|+|B[\tilde p](t)| < \infty.
$$
Using Lemma \ref{lem_cont}, we conclude that $\tilde{p}$ admits a continuation beyond $A_1$, a contradiction to our assumption $A_1 < \infty$. The solution therefore exists globally and it is straightforward to verify that a solution $\tilde p$ defines a $C^1$-velocity $v$ that solves \eqref{eq_mot}.
\end{proof}

\begin{remark}
Our Theorem \ref{thm_exist} shows the global existence of solutions for continuous external forces. It is natural to ask what happens if  $F_{ext}$ is zero for large times. One would hope that the particle attains an asymptotically constant velocity. However, due to the history-dependent self-fields, this appears to be a non-trivial question. We leave this issue open.
\end{remark}

\textbf{ACKNOWLEDGEMENTS:} The authors wish to cordially thank M. Kiessling and A. S. Tahvildar-Zadeh for stimulating discussions, providing a wealth of information on the electrodynamic self-force problem and in particular for inspiring our research by sharing their work \cite{KiesslingTahvildarZadeh18} in a preliminary form with us. Vu Hoang wishes to thank V. Perlick for very helpful discussions. The authors are also grateful to D. Hundertmark for remarks on the scattering problem and M. Hott for a general discussion on the subject. 

The work of Vu Hoang, Angel Harb, Aaron DeLeon and Alan Baza on this project was funded (full or in-part) by the University of Texas at San Antonio, Office of the Vice President for Research, Economic Development, and Knowledge Enterprise.  
Vu Hoang wishes to thank the National Science Foundation for support under grants DMS-1614797 and NSF DMS-1810687. 

%
%

\section{Appendix}


\subsection*{Appendix A: Convergence of the field tensor}

First we show the following useful property of a subluminal world-line:
\begin{proposition}\label{prop_lower_bound_D}
Let $\fq \in \SWL$. Then for any $(x^\al)\notin \fq$, 
\begin{align*}
    D(x^\al - q^\al(\tau'))^2 \geq (1-K^2)|\tau_{ret}(x^\al) - \tau'|^2
\end{align*}
for all $\tau' \leq \tau_{ret}(x^\al)$, where $K < 1$ is such that $|v(t)|\leq K$ for all $t \leq t_{ret}$ and $q^0(\tau_{ret}) = t_{ret}$.
\end{proposition}
\begin{proof}
Write $x^\al = (t, x), q^0(\tau_{ret}) = t_{ret}$ and $\mathbf{u} = (u^1, u^2, u^3)$. Then 
\begin{align}
    &D(x^\al - q^\al(\tau'))^2 = |x^0 - q^0(\tau')|^2 - |x - q(\tau)|^2\nonumber\\
    &= \left|x^0 - q^0(\tau_{ret}(x^\al)) + q^0(\tau_{ret}(x^\al)) - q^0(\tau')\right|^2 - \left|x - q(\tau_{ret}(x^\al)) +
    \int_{\tau'}^{\tau_{ret}(x^\al)} \mathbf{u}(\sigma)~d\sigma\right|^2\label{aux56}
\end{align}
Next we note that $\mathbf{u} = \gamma \mathbf{v}$ and that coordinate time and proper time differential are related by $dt = \gamma d\tau$. This allows us to estimate 
\begin{align*}
    \left|  \int_{\tau'}^{\tau_{ret}} \mathbf{u}(\sigma)~d\sigma \right| \leq \int_{\tau'}^{\tau_{ret}} |\mathbf{u}(\sigma)|~d\sigma \leq K  \int_{\tau'}^{\tau_{ret}} \gamma~d\sigma
     = K |q^0(\tau_{ret})-q^0(\tau')|.
\end{align*}
Returning to estimating $D(x^\al - q^\al(\tau'))^2$ and expanding \eqref{aux56} we get
\begin{align*}
    &D(x^\al - q^\al(\tau'))^2 \\
    &\geq |t-t_{ret}|^2 + 2|t-t_{ret}||q^0(\tau_{ret}(x^\al)) - q^0(\tau')| + |q^0(\tau_{ret}(x^\al)) - q^0(\tau')|^2 + \\
    & - |x - q(\tau_{ret}(x^\al))|^2 - 2 K |x - q(\tau_{ret}(x^\al))| |q^0(\tau_{ret}(x^\al)) - q^0(\tau')| - K^2 |q^0(\tau_{ret}(x^\al)) - q^0(\tau')|^2
\end{align*}
Finally we use $|t-t_{ret}| = |x - q(\tau_{ret}(x^\al))|$, which yields
\begin{align*}
    &D(x^\al - q^\al(\tau'))^2 \geq (1-K^2) |q^0(\tau_{ret}(x^\al)) - q^0(\tau')|^2\geq (1-K^2)|\tau_{ret}-\tau'|^2
\end{align*}
where we have used $\gamma \geq 1$. 
\end{proof}

\begin{proof}[Proof of Proposition \ref{prop_conv}.] 
To check convergence of \eqref{eq_pot}, we first note that by \eqref{eq_Kiess}
$$
\frac{|J(x)|}{x} \leq \frac{\text{const}}{(1+x^2)^{\frac{3}{4}}}.
$$
Let $K  < 1$ be such that $|v(\tau')|\leq K$ for all $\tau' \leq \tau_{ret}(x^\al)$. Using Proposition \ref{prop_lower_bound_D}, we can estimate the integrand of \eqref{eq_pot} by
\begin{align}
    \left|\frac{J_1(\ka D)}{D}u^a\right| \leq \frac{\text{const}(K, \ka)}{(1+(1-K^2) |\tau_{ret}(x^\al) -\tau'|^2)^{3/4}}
\end{align}
which implies absolute convergence of \eqref{eq_pot}. In order to justify the formal differentiation that leads to \eqref{eq_field_tensor}, it suffices by a standard result on parameter-dependent integrals to show that for every event $(x^\alpha)$ there exists a neighborhood of $(x^\alpha)$ and a function $M\geq 0, M\in L^1(-\infty, 0)$ so that
$$
\left|\frac{J_2(\ka D) R_{[c}u_{d]}}{\gamma D^2} \right|\leq M(t'-t)
$$
for all $y^\alpha$ in a neighborhood of $x^\alpha$. 
Each component of $R^a$ can be bounded by $|t-t'|$ and by using \eqref{eq_Kiess} again, we can take 
$$M(t) = \frac{\text{const}(K) |t|}{(1+|t|^2)^{5/4}}.$$ This yields 
\begin{align}\label{eq_devA}
\del_c A_a(x^\alpha) = \left.\frac{e \ka^2 R_c u_a}{2 S}\right|_{ret} + e\ka^2\int_{-\infty}^{\tau_{ret}(x^\alpha)} \frac{J_2(\ka D)}{D^2} R_c u_a(\tau)~d\tau   
\end{align}
and leads also directly to \eqref{eq_field_tensor}. Note that we have used $\left.\frac{J_1(x)}{x}\right|_{x=0} = \frac{1}{2}$.

Moreover, we have (recall \eqref{eq_someRel}) 
\begin{align}
\del_c S = u_c|_{ret} + (u_\alpha R^\alpha)^{-1} (1+a_\alpha R^\alpha) R_c|_{ret}
\end{align}
and the representation \eqref{eq_dev_field_tensor} follows formally from differentiating \eqref{eq_field_tensor}. By using similar estimates as before, the differentiation leading to \eqref{eq_dev_field_tensor} can be justified. Also note these estimates show that all the integrals in Proposition \ref{prop_conv} can be estimated by a constant $C(t_{ret})$ that does not depend on the specific $x^\al$, but only on $\tau_{ret}(x^\al)$.

To check that  \eqref{eq_field_tensor} defines a distributional solution of \eqref{eq_field} with distributional right-hand side \eqref{eq_current} we first note that $A^a=U^a-W^a$, where \eqref{eq_diff_two_wav} holds for $U^a, W^a$. A direct calculation shows that $\eqref{eq_field}$ holds in the sense of distributions, provided both $U^a$ and $W^a$ satisfy the Lorentz gauge condition 
\begin{align}\label{Lorentz_gauge}
\del_a U^a = \del_a W^a = 0
\end{align}
in the sense of distributions on $\R^4$. It is easy to check that the classical Li\'{e}nard-Wiechert potential $U^a$ satisfies \eqref{Lorentz_gauge}, so that in order to check \eqref{Lorentz_gauge} for $W^a$, we only need to check that the Lorentz gauge holds for $A^a$. This is done in the next Proposition.
\end{proof}

\begin{proposition}\label{prop_A_div}
For the vector potential defined by \eqref{eq_pot}, we have $\del_a A^a = 0$ in the distributional sense on $\R^4$.
\end{proposition}
\begin{proof}
From \eqref{eq_devA}, we see that for all $x^\alpha\in \R^4\setminus \fq$ 
\begin{align*}
\del_a A^a &= e\ka^2\left(\half + \int_{-\infty}^{\tau_{ret}(x^\alpha)} \frac{J_2(\ka D)}{D^2} R_b u^b(\tau)~d\tau \right)= e\ka^2\left(\half - \int_{-\infty}^{\tau_{ret}(x^\alpha)} \frac{\del}{\del\tau} \frac{J_1(\ka D)}{\ka D}~d\tau \right)= 0
\end{align*}
holds. So for the distributional derivative
\begin{align*}
\del_a A^a[\phi]=-A^a[\del_a \phi] = \int_{\R^4} \del_a A^a \phi~d^4 x + \lim_{\eps\to 0^+} \int_{W_\eps}A^a \phi~d\Sigma_{d} = 0
\end{align*}
where $W_\eps$ is a suitable world-tube surrounding the world-line (see \eqref{eq_World_tube}). Here, the limit of the integral over the world-tube's surface vanishes, since $A^a$ is uniformly bounded on the support of $\phi$.
\end{proof}


\subsection*{Appendix B: Energy-momentum tensor}

\begin{proof}[Proof of Proposition \ref{prop:EM_tens_cons}.] First we multiply the BLTP equation 
$(I-\ka^{-2}\wa) \del_\nu \Fumn = 4 \pi j^\mu $ by $F_{\al \mu}$ to get
\begin{align}
F_{\al \mu} \del_{\nu} \Fumn - \ka^{-2}  F_{\al \mu} \wa \del_\nu \Fumn = 4 \pi F_{\al \mu} j^\mu. \label{eq1} 
\end{align}
The idea is now to write the left hand side in the form of a total divergence $\del_\nu(-4\pi T_\al^\nu)$.
The following identity can be easily checked using the second equation of \eqref{eq_field}:
\begin{align}
F_{\al \mu} \del_\nu F^{\mu \nu} = \del_\nu [ F_{\al \mu} \Fumn + \frac{1}{4} \delta_{\al}^\nu \Fdrt \Furt]\label{ident1}
\end{align}
Using \eqref{ident1}, we can write the first term on the lhs of \eqref{eq1} as a divergence. The remainder of the calculation deals with the second term on the lhs of \eqref{eq1}.\\
\emph{Step 1:} We write 
\begin{align}
F_{\al \mu} \wa \del_\nu \Fumn &= \del_\nu[ F_{\al \mu} \wa F^{\mu \nu}] -(\del_\nu F_{\al \mu})\wa F^{\mu \nu}\nonumber\\
&= \del_\nu[ F_{\al \mu} \wa F^{\mu \nu}] + (\del_\al F_{\mu \nu}) \wa F^{\mu \nu} + (\del_\mu F_{\nu \al}) \wa F^{\mu \nu}\label{T1}
\end{align}
where we have used the second line of \eqref{eq_field} again.\\
\emph{Step 2:} We compute $(\del_\mu F_{\nu \al}) \wa F^{\mu \nu}$:
\begin{align}
(\del_\mu F_{\nu \al}) \wa F^{\mu \nu} &=(\del_{\nu} F_{\mu \al}) \wa F^{\nu \mu}= \del_\nu[F_{\mu\al}\wa F^{\nu\mu}]-F_{\mu\al}\del_\nu\wa F^{\nu\mu}\nonumber \\
&= \del_\nu\left[ F_{\mu \al}\wa F^{\nu \mu}\right] - F_{\al \mu} \del_\nu \wa F^{\mu \nu}\label{T2}
\end{align}
where we first renamed the contracting indices $\nu \leftrightarrow \mu$, used the reversed product rule and then used antisymmetry of $F$ to switch $\mu\al\to\al\mu$ and $\nu\mu\to\mu\nu$ in the second term. The second term of \eqref{T2} is
\begin{align}
F_{\al \mu} \del_\nu \wa F^{\mu \nu} &= \del_\nu[F_{\al \mu} \wa F^{\mu \nu}] - (\del_\nu F_{\al\mu}) \wa F^{\mu \nu}\nonumber\\
&=\del_\nu[F_{\al \mu} \wa F^{\mu \nu}] + (\del_\al F_{\mu\nu}) \wa F^{\mu \nu} + (\del_\mu F_{\nu\al}) \wa F^{\mu \nu}\label{T3}
\end{align}
where we have used the second line of \eqref{eq_field}. 
Inserting \eqref{T3} into \eqref{T2} gives:
\begin{align}\label{T4}
(\del_\mu F_{\nu \al} )\wa F^{\mu \nu} =  - \half (\del_\al F_{\mu \nu}) \wa F^{\mu \nu}.  
\end{align}
\emph{Step 3:} Up to now, we have obtained
\begin{align}\label{T5}
F_{\al \mu} \wa \del_\nu F^{\mu \nu} &=\del_\nu[F_{\al \mu} \wa F^{\mu \nu}] +\half (\del_\al F_{\mu \nu}) \wa F^{\mu \nu}. 
\end{align}
We now compute the remaining term $(\del_\al F_{\mu \nu}) \wa F^{\mu \nu}$:
\begin{align*}
(\del_\al F_{\mu \nu}) \wa F^{\mu \nu} = \del_{\al}[\Fdmn \wa \Fumn] - F_{\mu \nu} \wa \del_\al F^{\mu \nu}.
\end{align*}
Using again the second line of \eqref{eq_field} we observe that
\begin{align*}
F_{\mu \nu} \wa \del_\al F^{\mu \nu} 
&= F^{\mu \nu} \wa \del_\al F_{\mu \nu} = - F^{\mu \nu} \wa \del_\mu F_{\nu \al} - F^{\mu \nu} \wa \del_\nu F_{\al \mu}= - F^{\mu \nu} \wa \del_\mu F_{\nu \al} - F^{\nu \mu}\wa \del_\mu F_{\al \nu} \\
&= - 2 F^{\mu \nu} \wa \del_\mu F_{\nu \al}.
\end{align*}
Hence we can write 
$$
(\del_\al F_{\mu \nu})\wa F^{\mu \nu} = \del_\al[F_{\mu \nu} \wa F^{\mu \nu}] + 2 F^{\mu \nu} \wa \del_\mu F_{\nu \al}.
$$
Inserting this into \eqref{T5} we obtain
\begin{align}\label{T6}
F_{\al \mu} \wa \del_\nu F^{\mu \nu} &=
\del_\nu \left[ F_{\al \mu}\wa F^{\mu \nu} +  \half \delta_{\al}^\nu F_{\mu \rho} \wa F^{\mu \rho}\right] 
+ F^{\mu \nu} \wa \del_\mu F_{\nu \al}.
\end{align}
\emph{Step 4:} Next we deal with $F^{\mu \nu} \wa \del_\mu F_{\nu \al}$. First note
\begin{align}\label{T7}
F^{\mu\nu} \wa \del_\mu F_{\nu \al} &= F^{\nu \mu} \wa \del_\nu F_{\mu \al} 
= \del_\nu [F^{\nu \mu} \wa F_{\mu \al}] - (\del_\nu F^{\nu\mu}) \wa F_{\mu \al}
\end{align}
and observe that
\begin{align}
(\del_\nu F^{\nu \mu})\wa F_{\mu \alpha} &= (\del_\nu F^{\nu \mu}) \del^\beta \del_\beta F_{\mu \al} = -(\del_\nu F^{\nu \mu}) \del^\beta \del_\mu F_{\al \beta} - (\del_\nu F^{\nu \mu}) \del^\beta \del_\al F_{\beta \mu}\nonumber\\
&= -\del_\mu[(\del_\nu F^{\nu \mu})\del^\beta F_{\al \beta}] - (\del_\nu F^{\nu \mu}) \del^\beta \del_\al F_{\beta \mu}\label{T8}
\end{align}
where we have used that $\del_\mu \del_\nu F^{\nu \mu} = 0$ by antisymmetry of the field tensor. For the last term in the previous line we have
\begin{align*}
(\del_\nu F^{\nu \mu}) \del^\beta \del_\al F_{\beta \mu} =  \del_\al [(\del_\nu F^{\nu \mu}) \del^\beta F_{\beta \mu}] - (\del_\al \del_\nu F^{\nu \mu}) \del^\beta F_{\beta \mu}
\end{align*}
and $(\del_\al \del_\nu F^{\nu \mu}) \del^\beta F_{\beta \mu}$ can be written as $(\del_\al \del^{\beta} F_{\beta\mu}) \del_\nu F^{\nu \mu}$ and we have
\begin{align*}
  (\del_\nu F^{\nu \mu}) \del^\beta \del_\al F_{\beta \mu} &=  \half \del_\al[(\del_\nu F^{\nu\mu})\del^\beta F_{\beta\mu}].
\end{align*}
Hence \eqref{T8} is given by
\begin{align*}
(\del_\nu F^{\nu \mu}) \wa F_{\mu \al} = - \del_\nu [(\del_\mu F^{\mu \nu}) \del^\beta F_{\alpha \beta}] - \half \del_\al [(\del_\nu F^{\nu \mu}) \del^\beta F_{\beta \mu}] 
\end{align*}
which results \eqref{T7} in being
\begin{align}
    F^{\mu\nu} \wa \del_\mu F_{\nu \al} &= \del_\nu [F^{\nu \mu} \wa F_{\mu \al}] + \del_\nu [(\del_\mu F^{\mu \nu}) \del^\beta F_{\alpha \beta}] + \half \del_\nu\delta^\nu_\al [(\del_\rho F^{\rho \mu}) \del^\beta F_{\beta \mu}].
\end{align}
\emph{Step 5:} Finally, inserting the last result into \eqref{T6} we can write $F_{\al \mu} \wa \del_\nu F^{\mu \nu}$ as a total divergence:
\begin{align*}
F_{\al \mu} \wa \del_\nu F^{\mu \nu}&=
\del_\nu\left[ F_{\al \mu} \wa F^{\mu \nu} + \half \delta_{\al}^\nu F_{\mu \nu} \wa F^{\mu \nu}\right]\\
&\quad + \del_\nu\left[ F^{\nu\mu} \wa F_{\mu \al} + (\del_\mu F^{\mu\nu}) \del^{\beta} F_{\al \beta} + \half \delta_{\al}^\nu (\del_\rho F^{\rho \mu}) \del^\beta F_{\beta \mu}\right].
\end{align*}
Using this identity in \eqref{eq1}, we arrive at the statement of the proposition.
\end{proof}


\subsection*{Appendix C: An identity for Bessel functions}

In this appendix, we list and prove some properties of functions involving Bessel functions. 
First consider the following identity ($n=0, 1, 2, \ldots$):
\begin{align}
\frac{d}{dz}\frac{J_n(z)}{z^n} = - \frac{J_{n+1}}{z^n}.\label{besselderivative}
\end{align}
Recall that $J_n(z)$ is an entire function with the well-known power series representation 
$$
J_n(z) = \sum_{m =0 }^\infty \frac{(-1)^m}{m! (m+n)!} \left(\frac{z}{2}\right)^{2m + n}.
$$
We have
$$
\frac{J_n(z)}{z^n} = \sum_{m =0 }^\infty \frac{(-1)^m}{m! (m+n)!} \frac{z^{2m}}{2^{2m+n}}
$$
and by term-wise differentiation it follows that
\begin{align*}
\frac{d}{dz}\frac{J_n(z)}{z^n} &= \sum_{m=1}^\infty \frac{(-1)^m 2 m}{m! (m+n)!} \frac{z^{2m-1}}{2^{2m+n}}
= - \sum_{k=0}^\infty \frac{ (-1)^k (2k+2) }{k! (k+1) (k+n+1)!} \frac{z^{2 k +1 }}{ 2^{2k+1+n} 2} = -\frac{ J_{n+1}}{z^n}.
\end{align*}
The following Lemma estimates expressions of the form $J_2(x)/x^2$:
\begin{lemma}\label{lem_besseldifference}
For $x, y > 0$,
\begin{align}
    \left|\frac{J_2(x)}{x^2} - \frac{J_2(y)}{y^2}\right|\leq C\frac{|x-y|}{(1+\min\{x, y\}^2)^{5/4}}.
\end{align}
\end{lemma}
\begin{proof}
We apply the mean value theorem to the function $g(z) = J_2(z)/z^2$. Observe that $g'(z) = -J_3(z)/z^2$ by \eqref{besselderivative}. By \eqref{eq_Kiess}, 
\begin{align}
    |g'(z)|\leq \frac{C_3}{(1+z^2)^{7/4}}|z|\leq C(1+z^2)^{-5/4}
\end{align}
and hence
\begin{align*}
\left|\frac{J_2(x)}{x^2} - \frac{J_2(y)}{y^2}\right| &\leq C\frac{|x-y|}{(1+\min\{x, y\}^{2})^{5/4}}.
\end{align*}
\end{proof}


\subsection*{Appendix D: Differentiation of parameter-dependent integrals with singularities}

In the following Lemma, we differentiate a parameter-dependent integral that is mildly singular along a world-line. Since the location of the singularity depends on the parameter $\sigma$, the result is not quite standard and we provide a more detailed proof. 
\begin{lemma}\label{lem_dev_param}
Suppose that $\Omega\subset \R^4$ is a bounded region with Lipschitz boundary and that $Q^\al: \R\times (-a, a)\to \R^3$ is such that $\tau\mapsto Q^\al(\tau, \sigma)$ is a timelike curve for each $\sigma\in (-a, a)$. Suppose that $H(x, \sigma) : D \to \R$ is a function defined on $D = \{(x^\al, \sigma) \in \R^4\times (-a, a): x^\al \neq Q^\al(\tau, \sigma),\tau\in(-\infty,\infty)\}$ and that the following holds:
\begin{enumerate}
    \item $H$ is smooth on $D$ and has a compact support in $x^\al$ for each $\sigma$. Moreover, $\del_\sigma H$ is continuous at every point $(x^\al,\sigma)\in D$.
    \item There exists a function $\Phi: \R\times (0, \infty) \to \R^+$, $\Phi=\Phi(\tau,r)$,  such that $r^2 \Phi\in L^1( \R\times (0, \infty))$, independent of $\sigma$, and such that
    \begin{align*}
        |H(\cdot, \sigma)\circ X^\al_\sigma| + |\del_\sigma H(\cdot, \sigma)\circ X^\al_\sigma|  \leq \Phi(\tau, r)
    \end{align*}
    where $X^\al_\sigma(\tau, r, \theta, \varphi)$ is the coordinate transformation into retarded light-cone coordinates attached to the world-line $Q^\al(\cdot, \sigma)$. 
    \item $H(\cdot, \sigma) \circ X^\al = H_0(\tau, \sigma) r^{-2} + H_1(\tau, r, \theta, \varphi, \sigma)$ with
    \begin{align*}
        |H_1(\tau,r, \theta, \varphi, \sigma)|\leq \Phi_1(\tau)r^{-1}
    \end{align*}
    and $H_0$ a smooth function and $\Phi_1\in L^1(\R)$. 
\end{enumerate}
Then 
\begin{align}
    \del_\sigma \int_{\R^4} H(x, \sigma) \sqrt{-g}~dx = \int_{\R^4} \del_\sigma H(x, \sigma) \sqrt{-g}~dx.
\end{align}
\end{lemma}
\begin{proof}
We first note that the retarded light-cone coordinate mapping $X_\sigma(\tau, r, \theta, \varphi)$ 
is a smooth diffeomorphism of $R = \R\times (0, \infty)\times \sph$ in and onto $\R^4 \setminus \{Q(\cdot, \sigma)\}$ for each $\sigma$. We change from cartesian coordinates into retarded light-cone coordinates by $(x^\al)\mapsto (y^{\al'}), y^{\al'} = (\tau, r, \theta, \varphi)$.

Define the vector field $V$ by 
\begin{align*}
    V^\al(x^\mu) := \left.\del_\sigma X^\al_\sigma\right|_{X^{-1}_\sigma(y^\mu)}
\end{align*}
where $V^\al$ are the components of the vector field with respect to a standard cartesian Minkowski coordinate system. In the $(y^{\al'})$ coordinate system the components of $V$ are given by 
\begin{align*}
    V^{\al'}(y^\mu) = \left.\frac{\del y^{\al'}}{\del x^\al }\del_\sigma X^\al_\sigma\right|_{X^{-1}_\sigma(y^\mu)} =
    \left.\frac{\del (X^{-1}_\sigma)^{\al'}}{\del x^\al} \del_\sigma X^\al_\sigma\right|_{X^{-1}_\sigma(y^\mu)}.
\end{align*}
We first write $I(\sigma) := \int_{\R^4} H(x, \sigma) \sqrt{-g}~dx = I_{1,\eps}(\sigma)+I_{2,\eps}(\sigma)$ where
\begin{align*}
    I_{1,\eps}(\sigma) =\int_{\{X^\al_\sigma(\tau, r,\theta, \varphi) : r \geq \eps \}} H \sqrt{-g}~d x
\end{align*}
and study the differentiability of $I_{1,\eps}$.
By changing coordinates via the mapping $X_\sigma$ we get
\begin{align*}
    I_{1,\eps}(\sigma) = \int_{R_\eps} H\circ X_\sigma ~\sqrt{-g}~dy
\end{align*}
with the $g$ being the metric determinant in the $(y^{\al'})$-coordinates and $R_\eps = \R\times [\eps, \infty)\times \sph, \sqrt{-g}~dy = r^2 \sin \theta~ d\tau~ dr~ d\theta~ d\varphi$.
We will now use standard results on differentiating parameter-dependent integrals to interchange integration and $\del_\sigma$. First note that $\sigma \mapsto H\circ X_\sigma(y)$ is continuously differentiable on $(-a, a)$ for fixed $y\in R = \R\times [0,\infty)\times \sph$ and that by the chain rule
\begin{align*}
    \del_\sigma (H\circ X_\sigma) = \del_\sigma X^\mu_\sigma \del_\mu H \circ X_\sigma  + \del_\sigma H \circ X_\sigma
\end{align*}
on $R_\eps$.
A calculation using $\del_\sigma g = g^{\al\beta} \del_\sigma g_{\al\beta}$ gives 
\begin{align*}
    \del_\sigma I_{1,\eps}(\sigma) &= \int_{R_\eps} \left[\del_\sigma H \circ X_\sigma +  \nabla_{\al'}( (H V^{\al'}) \circ X_\sigma )\right] ~~\sqrt{-g}~dy\\
    &=\int_{\{X^\al_\sigma(\tau, r,\theta, \varphi) : r \geq \eps \}} \left[\del_\sigma H  +  \nabla_{\al} (H V^{\al})\right] ~~\sqrt{-g}~dx.
\end{align*}
Here, we changed back to $(x^\al)$-coordinates. Next, we consider the difference quotient
\begin{align*}
    &\frac{I(\sigma)-I(0)}{\sigma} = \frac{I_{1,\eps}(\sigma)-I_{1,\eps}(0)}{\sigma}+
    \frac{I_{2,\eps}(\sigma)-I_{2,\eps}(0)}{\sigma}
     = \del_\sigma I_{1,\eps}(\eta(\eps, \sigma)) + \frac{I_{2,\eps}(\sigma)-I_{2,\eps}(0)}{\sigma}
\end{align*}
where $|\eta(\eps, \sigma)|\leq \sigma$. By Assumption 2, for given fixed $\sigma$, $\eps=\eps(\sigma)>0$ can be chosen so small such that $|I_{2,\eps}(\sigma)|+|I_{2,\eps}(0)|\leq \sigma^2$. Next we estimate $\del_\sigma I_{1,\eps}(\eta(\eps, \sigma))$. 
Note that $\nabla_\al (H V^\al) = \frac{1}{\sqrt{-g}} \del_\mu (H V^\mu \sqrt{-g})$ by a well-known formula. Using integration by parts,
\begin{align*}
&\int_{\{X^\al_\eta(\tau, r,\theta, \varphi) : r \geq \eps \}} \nabla_{\al} \left.(H  V^{\al})\right|_{\sigma=\eta} ~~\sqrt{-g}~dx = - \int_{\del \{X^\al_\eta(\tau, r,\theta, \varphi) : r \geq \eps \}} \left.H V^\al\right|_{\sigma=\eta}~d\Sigma_\al \\
&= -\int_{\R}\int_{\sph} \left(H_0(\tau, \eta)\eps^{-2} + H_1(\tau, \eps, \theta, \varphi, \eta)\right)V^\al N_\al~\eps^2 \sin \theta~d\tau d\theta d\varphi \\
&= -\int_{\R} H_0(\tau, \eta) \int_{\sph}  V^\al N_\al~ \sin \theta~d\tau d\theta d\varphi - \int_{\R}\int_{\sph} H_1(\tau, \eps, \theta, \varphi, \eta) V^\al N_\al~\eps^2 \sin \theta~d\tau d\theta d\varphi
\end{align*}
Now note that $V^\al = \del_\sigma X^\al = \del_\sigma Q^\al + r (\del_\sigma u^\al + \del_\sigma N^\al) = \xi^\al(\tau, \eta) + r (\del_\tau \xi^\al + \del_\sigma N_\al)$ which means that $V^\al = \xi^\al + \calO(\eps)$. Note moreover that
\begin{align*}
    \int_{\sph}  N_\al(\theta, \varphi, \sigma) \sin\theta\, d\theta d\varphi = 0 
\end{align*}
for all $\sigma$. Using Assumption 3, 
\begin{align*}
    \int_{\{X^\al_\eta(\tau, r,\theta, \varphi) : r \geq \eps \}} \nabla_{\al} \left.(H  V^{\al})\right|_{\sigma=\eta} ~~\sqrt{-g}~dx = \calO(\eps)
\end{align*}
uniformly for all $\eta$ in a small, fixed neigborhood of $0$. Also,
\begin{align*}
    \int_{\{X^\al_\sigma(\tau, r,\theta, \varphi) : r \geq \eps \}} \del_\sigma H ~\sqrt{-g}~dx = \int_{\R^4} \del_\sigma H(x, 0) ~\sqrt{-g}~dx + \calO(\eps)
\end{align*}
uniformly for all $\eta$ in a small, fixed neigborhood of $0$. Therefore,
\begin{align*}
    \lim_{\sigma \to 0}\frac{I(\sigma)-I(0)}{\sigma} = \int_{\R^4} \del_\sigma H(x, 0) ~\sqrt{-g}~dx 
\end{align*}
as desired. 
\end{proof}


\Addresses

\end{document}